\documentclass[reqno,11pt,a4paper]{amsart}
\usepackage[margin=2.5cm]{geometry}

\newcommand*{\medcap}{\mathbin{\scalebox{0.9}{\ensuremath{\bigcap}}}}
\newcommand{\QcirP}{\operatorname{Q\mkern-2mu\odot\mkern-2mu P}}
\newcommand{\subQcirP}{\operatorname{Q \odot P}}

\newcommand{\TT}{\mathbb{T}}

\usepackage{natbib}
\usepackage{amsmath, amsthm}
\usepackage{amsfonts,amssymb,dsfont}
\usepackage{nicefrac,mathrsfs}
\usepackage{bm}                                 
\usepackage[backgroundcolor=white,bordercolor=orange]{todonotes}

\renewcommand{\P}{\mathbb{P}}

\newcommand{\ccF}{{\mathscr F}}

\newcommand{\ccH}{{\mathscr H}}

\newcommand{\ccM}{{\mathscr M}}

\newcommand{\cZ}{{\mathcal Z}}

\DeclareMathOperator{\Var}{Var}

\renewcommand{\P}{\mathbb{P}}

\newcommand{\Ind}{{\mathds 1}}
\newcommand{\ind}[1]{\Ind_{\{#1\}}}

\newcommand{\RR}{\mathbb{R}}
\newcommand{\FF}{\mathbb{F}}
\newcommand{\bbF}{\mathbb{F}}
\newcommand{\bbH}{\mathbb{H}}
\newcommand{\bbG}{\mathbb{G}}

\renewcommand{\P}{P} 

\newenvironment{enumeratei}
  {\begin{enumerate} }
  {\end{enumerate}}

\theoremstyle{theorem}
\newtheorem{theorem}{Theorem}[section]
\newtheorem{corollary}[theorem]{Corollary}      
\newtheorem{proposition}[theorem]{Proposition}  

\theoremstyle{definition}
\newtheorem{example}{Example}[section]

\newtheorem{remark}{Remark}[section]

\newtheorem{assumption}{Assumption}[section]

\RequirePackage{colortbl}
\definecolor{tscolor}{rgb}{1.0,0.6,0.0}
\usepackage[ulem=normalem,final]{changes}      %
\definechangesauthor[name = Thorsten Schmidt, color=tscolor]{TS}  %
\usepackage[colorlinks,urlcolor=red,citecolor=blue,linkcolor=red]{hyperref}

\usepackage{enumerate}

\newcommand{\SB}{{\rm SB}}

\renewcommand{\mid}{\, \big| \,}

\begin{document}

\title{Insurance products with guarantees in an affine setting}

\author{Raquel M. Gaspar}
\address{ISEG Research, Lisbon School of Economics and Management, Universidade de Lisboa, Portugal}
\email{rmgaspar@iseg.ulisboa.pt}

\author{Thorsten Schmidt}
\address{Department of Mathematical Stochastics, University of Freiburg, Germany}
\email{thorsten.schmidt@stochastik.uni-freiburg.de}

\thanks{The work of R.M.~Gaspar was partially supported by FCT, I.P., the Portuguese national funding agency for science, research and technology, under the Project UID06522. We also thank the MAPFRE foundation for support through the Research Grant Ignacio H. de Larramendi. The work of T. Schmidt was partially supported by a grant from Deutsche Forschungsgemeinschaft under the project SCHM 2160/15-1. Support of the FDMAI is gratefully acknowledged. We want to thank Wilfried Donatien Kuissi Kamdem and Felix Tambe Ndonfack for a careful reading of the manuscript and David Criens for his comments on an earlier version of the manuscript.}
\maketitle

\begin{abstract}
For the attractivity of medium- and long-term insurance products
it is necessary to participate on the profitability of stock markets. 
To eliminate downside-risk, guarantees should be included  which naturally gives rise to the problem of valuing contracts in a unified insurance–finance framework.

We study a general setup that allows for joint modelling of financial markets, mortality, and policyholder behaviour. Within this framework, we propose a general affine approach and obtain explicit valuation formulas for variable annuities and related contracts that remain computationally tractable due to the affine structure. The model permits flexible dependence between mortality and equity dynamics, as highlighted by the empirical evidence from the COVID-19 pandemic. Moreover, surrender intensities are modelled as functions of the driving affine process, thereby introducing market dependence into lapse behaviour. The resulting framework combines analytical tractability with sufficient flexibility to capture key features of long-term insurance products. 
\end{abstract}

\vspace{2mm}

\section{Introduction}

The design of retirement products, and more generally medium- and long-term insurance products, remains one of the central challenges  in modern insurance mathematics since longevity, low interest rates and uncertainty about asset returns must be handled in an cost-efficient and risk-sensitive manner. Variable annuities and related products form an important class in this context, combining participation in equity markets with embedded guarantees to reduce risk. Their valuation and risk-management poses significant challenges, in particular over a long time horizon. 

We start in a general setup in discrete time and show how to obtain valuation formulas that guarantee absence of insurance-finance arbitrage in the sense of \cite{artzner2024}. To obtain tractable results we use enlargement of filtration techniques and derive quite general results for one, two or possibly more stopping times related to the contract (like mortality and surrender). 

Towards a flexible and tractable setting, we propose to use affine models. Affine processes  constitute a flexible and highly tractable class which makes them well suited for the problem at hand. 
In contrast to many existing approaches we allow for a quite general dependence between insurance  and finance products, which leads to technical difficulties. In particular, we address the valuation problem when both mortality and surrender are possibly correlated to the financial markets. We propose to use enlargement-of-filtration techniques which are, to the best of our knowledge,  applied in this form for the first time to the valuation of insurance products. The considered framework  also allows to incorporate more than two stopping times. 

In the literature,  different approaches have been proposed for the valuation of variable annuities and similar products (see e.g., \cite{dhaene2017fair}, \cite{Dhaenekukushetal-2013}, \cite{pelsser_stadje_2014} and  \cite{semyon_eugene_wuerthich_2008} and references therein). More recently, 2-step and 3-step approaches have been proposed as for example in \cite{deelstra2020valuation, barigou2023actuarial} and in \cite{linders23}. Extensions of the insurance-finance framework we consider here to uncertainty can be found in  \cite{oberpriller2024robust} and to the benchmark approach in \cite{PlatenSchmidtSchmutz25}. 

On the other side, tontines and related products have been proposed as competitive pension products and we refer to \cite{milevsky2015optimal}, \cite{winter2022modern}, \cite{chen2020optimal},  \cite{chen2019tonuity} and further literature therein.

The paper is organised as follows: in Section \ref{sec:IFA} we introduce insurance-finance markets and provide with Theorem \ref{thm-FTIFA} an fundamental theorem which will be used later on for valuation without insurance-finance arbitrage, a so-called IFA-free valuation. In Section \ref{sec:IFA VA} we introduce the framework for insurance products with guarantees, which include variable annuities as a special case. Proposition \ref{prop:VA} as a key result provides IFA-free valuation formulas for all building blocks of a variable annuity with guarantee, surrender and death benefit. For more explicit and tractable formulas we will use enlargements of filtrations, which are introduced in section \ref{sec: progressive}. In particular, results for multiple stopping times are given and the setting for two stopping times (mortality and surrender) is developed in more detail. Section \ref{sec:affine} introduces a general affine process under the statistical and the risk-neutral measure and develops the key formulas needed for the valuation. Section \ref{sec:valuation} provides the valuation results for the building blocks of a variable annuity in an affine framework and Section \ref{sec:conclusion} concludes.

\section{Insurance-finance markets and arbitrage-free valuation}
\label{sec:IFA}
We follow \cite{artzner2024} and consider financial markets and insurance contracts with minimal assumptions on their dependence. To this end, consider a probability space $(\Omega, \ccH,P)$ and a discrete, finite time interval $\TT=\{0, 1, ...,T\}$. Information is divided into publicly available information (like stock prices, life tables, etc.), captured by the filtration 
$\bbF=(\ccF_t)_{t\in\TT}$ and internal information only available to the considered insurance company. This are typically the survival times of the insured clients, together with further information on the clients, like for example health states, encoded in the filtration  $\bbH=(\ccH_t)_{t\in\TT}$. We assume that $\bbH$ already encompasses public information , i.e.
\begin{align*}
 \ccF_t \subseteq\ccH_t \quad \text{for } t=0,\dots,T.
\end{align*}

\begin{remark}[Choice of the filtrations]
The main role of $\bbF$ is to capture information which does \emph{not} introduce arbitrage on the financial market or changes the market pricing, see Theorem \ref{thm-FTIFA} for the precise statement. Hence, one could also include insurance information here, as long as no financial arbitrage is created, which would both simplify pricing and hedging. 
In contrast, the role of $\bbH$ is to be general enough to capture insurance information which possibly could lead to an arbitrage on the financial market. Using this information is forbidden, however, by the law of insider trading.  
In \cite{artzner2024}, an additional filtration $\bbG$ was introduced to precisely capture this effect, while here - for simplicity - we stay with two filtrations only.
\end{remark}

\subsection{The financial market}\label{ssec-finmar}
The financial market consists of the risk-free account and $d$ tradeable securities.  Discounted price processes of the tradeable securities are given by the  $\FF$-adapted  process $ S=(  S^1,\dots,  S^d)$.

A \emph{trading strategy} on the financial market is  a $d$-dimensional, $\bbF$-adapted process $\xi=(\xi_t)_{0 \le t \le  T-1}$ with $\xi_t=(\xi^1_t,\dots,\xi^d_t)$. For a trading strategy, its \emph{(discounted) value}
is given by (see for example \cite{FoellmerSchied}, Proposition 5.7),
 $$ V^F(\xi) :=  (\xi \cdot S)_T =
  \sum_{t=0}^{T-1} \sum_{j=1}^d \xi^j_t \; \Delta S^j_{t}, $$
 with $\Delta S_{t} = S_{t+1} - S_{t}$.
Here we choose to add the superscript $F$ to emphasize that $V$ is the value of trading on the financial market.

 It is well-known that the financial market is arbitrage free if there exists an equivalent martingale measure. For the following we will therefore assume that  the set $\ccM_{e,b}(S,\FF)$ of equivalent martingale measures with bounded densities is not empty.
\subsection{The insurance market}\label{ssec-stacon}

Besides trading on the financial market, the insurance company can build a portfolio of insurance contracts. We assume that the insurance company can contract with possibly infinitely many \emph{insurance seekers}.

We concentrate on one type of insurance contract.
Such a contract, initiated at time $t$, offers \emph{benefits} at maturity $T$ which can be identified with a $\ccH_T$-measurable non-negative random variable $X_{t,T}$ (already discounted). In exchange for this benefits, the insured pays a non-negative premium $p_t$ at $t$, satisfying 
\begin{align}\label{cond-pt}
  p_t \in L^1_+(\ccH_t)=L^1_+(\Omega,\ccH_t,P).
\end{align}

While  financial contracts are standardised, insurance contracts are individual: they are linked to personal quantities as for example the life time of the insured, such that each contract is different.

We consider a homogeneous cohort (say of the same gender, age and health state) which buy the same type of insurance contract and therefore will pay the same premium for the contract. 
The associated $\ccH_T$-measurable benefits of the individual contracts, however, are different and are denoted by $X^1_{t,T},X^2_{t,T},\dots\; $. 
With these, insurance trading at $t$ is described by  an  \emph{insurance allocation} $\psi=(\psi^{i}_t)_{t \ge 0, i \ge 1}$: for each $t =0,\dots,T-1$, this  is a  $\ccH_t$-measurable, non-negative random sequence, where $\psi^{i}_t$ denotes the size of the contract with the $i^{th}$ insurance seeker.
The (discounted) value of the allocation is hence given by 
\begin{align}
        \label{VItpsi}
        V^I(\psi):= \sum_{t=0}^{T-1} \sum_{ i \ge 1} \psi^{i}_t \big( p_t - X^i_{t,T} \big),
    \end{align}
    where $I$ is used as superscript to describe the value of building an insurance allocation. 
    
 To obtain realistic strategies, we assume that in an allocation the insurance is allowed to trade only with finitely many contracts, whereafter we take limits. More precisely,  an \emph{insurance portfolio strategy} is a  sequence $\psi:=(\psi^{n})_{n \ge 1}$ of allocations. Each allocation $\psi^n=(\psi_t^{n,i})_{t \ge 0, i \ge 1}$ has only finitely many non-negative entries. In addition, we impose the following \emph{admissibility condition} for a portfolio strategy:
\emph{Convergence of the insurance volume}: there exist random variables $\gamma_t \ge  0 $, $0\le t<T$ so that
\begin{align} \label{conv-psi}
            \parallel\psi^n_t\parallel:= \sum_{i \ge 1}~\psi^{n,i}_t  \rightarrow \gamma_t \quad \text{a.s. for all } t<T.
\end{align}
The precise measurability of $\gamma_t$ is explained in the next subsection.
Later, we are interested in the particular case of a bounded portfolio strategy:
in addition to \eqref{conv-psi} we say that the  portfolio strategy $\psi$  is \emph{bounded} if there exists $c>0$ so that
 \begin{align} \label{bdd-psi}
            \parallel \psi^n_t\parallel \le c.
        \end{align}
for all $n\ge 1$ and $0\le t < T$.

It turns out that the $\sigma$-algebra
\begin{align}\label{de-GtT}
 \ccH_{t,T} := \ccH_t \vee \ccF_T
\end{align}
containing the insurance information up to date $t$ and $\ccF_T$, plays a distinctive role.
We make the following assumptions: \begin{assumption}\label{AssX1}
    For all $t \in \TT$, the standard contract $X_{t,T}\in L^2(\Omega,\ccH_T,P)$ and the individual ones $X^i_{t,T}$ satisfy
    \begin{enumerate}
        \item $  X^1_{t,T}, X_{t,T}^2,\dots \in L^2(\Omega,\ccH_T,P)$ are  $\ccH_{t,T}$-conditionally independent,
        \item $ \ E[ X^i_{t,T}|\ccH_{t,T} ] = E[X_{t,T} | \ccH_{t,T}], \ i=1,2,\dots$, and
        \item $\Var(X^i_{t,T}|\ccH_{t,T}) = \Var(X_{t,T}|\ccH_{t,T}) < \infty,  \ i = 1,2,\dots.$
    \end{enumerate}
\end{assumption}

Intuitively, the assumption on \emph{conditional} independence is very general since it includes all available financial information (until the final time $T$). This covers in particular the important case of a pandemic: when stock prices fall and mortality increases - which can be modelled through a hidden factor or via an intensity directly depending on the stock market, as for example in \cite{ballotta2019variable}.

An \emph{insurance-finance strategy} is now the pair $(\psi,\xi)$ which achieves the (discounted) \emph{insurance-finance value}
\begin{align}\label{IFR}
\lim_{n \to \infty}V^I(\psi^n)+ V^F(\xi).
\end{align}

\subsection{Insurance-finance arbitrage}\label{ssec-infiarb}

The \emph{insurance-finance market} is hence given by the triplet $(X,p,S)$.
An admissible insurance portfolio strategy $(\psi^n)_{n \ge 1}$, and an insurer's trading strategy $\xi$ form an \emph{insurance-finance arbitrage} (IFA), if
\begin{align}\label{con-IFarb}
   \lim_{n \to \infty} V^I(\psi^n)+ V^F(\xi) \in L_0^+ \backslash \{0\}.
    \end{align}

If there is no general IFA on the insurance-finance market, we say  no general insurance-finance arbitrage (NIFA${}^0$) holds. If there is   no bounded IFA, we say  NIFA${}^\infty$ holds.

In the following theorem, contained in  \cite{artzner2024}, 
absence of insurance arbitrage is characterized, however in a slightly simplified version.
For this result we  also rely on  a measure $P^*$ which is equivalent to $P$ (and later on take the role of an equivalent martingale measure, i.e.\ a risk-neutral measure). 
We  need a similar assumption as Assumption \ref{AssX1} for $P^*$, where we additionally  assume that the conditional expectation of $X_{t,T}$ under $P^*$ coincides with those under $P$: 

\begin{assumption}\label{AssX1*}
Consider $P^* \sim P$ and assume that 
   for all $t \in \TT$, 
    \begin{enumerate}
        \item $  X^1_{t,T}, X_{t,T}^2,\dots \in L^2(\Omega,\ccH_T,P^*)$ are  $\ccH_{t,T}$-conditionally independent under $P^*$,
        \item $ \ E_{P^*}[ X^i_{t,T}|\ccH_{t,T} ] = E_{P^*}[X^1_{t,T} | \ccH_{t,T}], \ i=2,3,\dots$, and
        \item $\Var_{P^*}(X^i_{t,T}|\ccH_{t,T}) = \Var_{P^*}(X^1_{t,T}|\ccH_{t,T})< \infty,  \ i = 2,3,\dots.$
    \end{enumerate}
\end{assumption}

\renewcommand{\theenumi}{\roman{enumi}}
\renewcommand{\labelenumii}{(\theenumi.\theenumii)}
\makeatletter\renewcommand{\p@enumii}{\theenumi.}
\makeatother

We recall that $\ccM_{e,b}(S,\bbF)$ is the set of equivalent martingale measures with bounded density on the financial market given by stock prices $S$ and filtration $\bbF$. Then, the following holds.

\begin{theorem}\label{thm-FTIFA}
On the insurance-finance market $(X,p,S)$ with Assumption \ref{AssX1},  the sequence of implications
 \eqref{FTIFA1}$\Rightarrow$
\eqref{FTIFA3} \; holds for the following assertions:
  \begin{enumerate}
    \item\label{FTIFA1}\qquad $NIFA^0$ holds,
    \item\label{FTIFA3}\qquad There exists  $P^*\sim P$ on $(\Omega,\ccH_{T-1,T})$ so that
\begin{enumerate}
  \item\label{coFT31}
  $\; P^*|_{ \ccF_T} \in \ccM_{e,b}(S,\bbF)\;$ and
  \medskip
  \item\label{coFT32}$\; E_{P^*}\big[p_t-X_{t,T}\big|\ccF_t\big] \le 0$ for $t=0,\dots,T-1$.
\end{enumerate}
  \end{enumerate}
  Moreover, if \eqref{FTIFA3} holds and $P^*$ satisfies Assumption \ref{AssX1*}, then 
 \begin{enumerate}
 	\setcounter{enumi}{2}
 	\item  \label{FTIFA4}\qquad $NIFA^\infty$ holds.\end{enumerate}
\end{theorem}

For us, the most important implication is  \eqref{FTIFA4}: if we find such a $P^*$ then there are no insurance-finance arbitrages. This can be achieved as follows: assume that  Assumption 2.1 holds and denote by $L$ the Radon-Nikodym derivative of $Q$ with respect to $\P|_{\ccF_T}$, such that $dQ=L\; dP$. Then,  define $P^*$ on $(\Omega, \ccH_T)$ by
\begin{align} \label{def:QcP}
            dP^* = L\, dP.
\end{align}
In this case, $P^*$ is the so-called \emph{QP-measure} which we denote by $\QcirP$. By Proposition 4.1 in \cite{artzner2024}, this $P^*$ coincides with $Q$ on $\ccF_T$, hence it is \emph{market-consistent}. On the other side,  for $X \in L^1$ or bounded from below,
        \begin{align}\label{PQ-rule:filtation}
          E_{\subQcirP}[X \,|\, \ccF_t]  & = E_Q \big[ E_P [ X \,|\,\ccF_T ]\, |\,\ccF_t \big].
    \end{align}
 Intuitively, this rule describes how to evaluate insurance contracts in an IFA-free way: by projection onto the publicly available information (under $P$) and afterwards by applying the risk-neutral pricing rule (under some $Q$ calibrated to available market data).  For further details we refer to \cite{artzner2024}.
    
In the following, we will show how to use this result on highly flexible affine models in insurance-finance markets.

\bigskip

\section{IFA-free valuation of  insurance products with guarantees}
\label{sec:IFA VA}

In this setting, we provide a general framework for insurance products with guarantees, where we include typical contract specifications of variable annuities. The setting is modular and  general, such that also other types of insurance products can be valued using this approach. 

A \emph{variable annuity} (VA) is an insurance contract which gives the holder a variety of benefits in exchange for a sequence of payments. The contract  has a maturity $T>0$. 
We start by explaining typical contract specifications of a variable annuity. 

While up to now we considered already discounted quantities, for the contract specifications we prefer to state the quantities  in undiscounted  terms, and denote by 
$$ \beta(t,T)$$ 
the discounting factor at time $t$ of the secure payment of 1 unit of money at time $T$ (which is often captured through the bank account $S^0$ with $\beta(t,T) = \nicefrac{S^0_t}{S^0_T}$). Whenever starting from time $0$, we simplify the notation by setting $\beta(T):=\beta(0,T)$.

\subsection{Contract details}\label{sec:contract details}
The payments $\pi_1,\dots,\pi_n$ are due at  a discrete time grid $T_0=0<T_1<\cdots<T_n<T_{n+1}=T$, where the annuity requires the payment $\pi_i$ at time $T_i$. These investments are immediately invested in a fund and we denote the value process of the fund by $F^\pi=(F^\pi_t)$. For simplicity, we assume that $\beta$ is deterministic, since in most practical applications we have in mind the impact of the other quantities (mortality, stock price, growth rate) will be significantly stronger.

\begin{remark}[On the underlying fund]
The main observation due in this respect is that the premium is not paid at a single time up-front, but by a series of payments. This means, if investments are in the stock $S^1=S$, say, that the value of the underlying fund at time $t$ is, given no surrender or death, 
\begin{align}\label{eq:F pi} 
    F_t^\pi = \Big( \sum_{i=1}^n \ind{T_i \le t} \frac{\pi_i}{S_{T_i}} \Big) \cdot S_t.
\end{align}
Consequences of these types of payments where already studied in \cite{bernard2017impact} and in \cite{hainaut2025participating}, but not in a general affine framework.    
\end{remark}

\bigskip

In the specification which we consider here, the VA offers three features:
\begin{enumerate}[(i)]
\item a guaranteed minimum accumulation benefit (GMAB),
\item a surrender benefit (SB), and 
\item a death benefit (DB).
\end{enumerate}
The precise description of these benefits can be done as follows: 
first, the GMAB equals the best of two investment choices, either the value of the underlying fund or a guaranteed amount. 

More specific, let us consider the guarantee of a fixed interest rate $\delta$ over the lifetime of the contract: if we upcount all payments to the maturity $T$, we obtain that the guarantee sums up to 
$$ K^\pi_T = \sum_{i=1}^n e^{\delta(T-T_i)} \pi_i. $$
Analogously, the value of the accumulated payments done until time $t$ computes to  $ K^\pi_t = \sum_{i=1}^n \ind{T_i \le t} e^{\delta(t-T_i)} \pi_i$.

At maturity the policyholder receives either the value of the fund or the value of the payments, invested with the guaranteed interest rate $\delta$, which equals \begin{align} \label{eq:GMAB}
    \max (  F_T^\pi,  K_T^\pi ). 
\end{align}
This, however, can only be claimed if the policyholder is still alive at time $T$ \emph{and} she did not surrender before.
  
Second, in case of early \emph{surrender} 
her right of refund is restricted to the current fund account value reduced by a \emph{compulsory surrender penalty}. 
For simplicity, we assume surrender is possible on the same grid as the payments are done. In case of surrender at time $T_i$, the surrender benefit is given by 
\begin{align} \label{eq:SB}
   \SB(T_i):= F^\pi_{T_i} \cdot p_{\text{SB}}(T_i), 
\end{align}
where the penalty $p_{\text{SB}}:[0,T] \to (0,1]$ is a deterministic function of time. Choosing $p$ small at early times and equal or close to one at later times  allows to overcome the problem that initial expenses of the contract may not be covered in the case of early surrender.

Third, in case of death before $T$, the \emph{death benefit} provides a payoff, but only  if the early surrender option has not been exercised yet. It includes the same guarantee (for simplicity) as in the GMAB, but this time paid out at the next $T_i$ after death. Therefore, the payoff of the death benefit in case of death at time $t$ such that $T_{i-1} < t \leq T_i$ is settled at time $T_i$ and is given by 
  \begin{align} \label{eq:DB}
     {\rm DB}(T_i) :=  \max(F_{T_i}^\pi, K_{T_i}^\pi).
  \end{align}

  \subsection{Mortality and surrender}
  Regarding mortality and surrender we assume that the given insurance clients are a homogeneous cohort in the sense that the probabilistic description of the individual mortality time $\tau^{m,i}$ can be captured by the probabilistic modelling of the mortality $\tau^m$ of a typical individual, which we detail in the following. This is done via an intensity-based approach and hence depends on the individual survival $\tau^{m,i}$, but otherwise on a common intensity. Furthermore, as typical in life-insurance, we introduce the current age $x$ as additional covariable.
  
In this regard,  assume that the remaining lifetime $\tau^m(x)$ of a insured person aged $x$ years is a $\bbH$-stopping time with %
intensity $\lambda^m(x)$. i.e. 
   \begin{align} \label{def:taum}
      P(\tau^m(x) > T \,|\,\ccH_t, \,\tau^m(x)>t) = \ind{\tau^m(x) > t}E\Big[e^{-\int_t^T \lambda^m_s(x) ds}|\ccH_t\Big].
  \end{align}

  \begin{remark}[Intensity-based setting]
  In particular, for each individual $i$ aged $x$ years we have that
   \begin{align} 
      P(\tau^{m,i}(x) > T \,|\, \ccH_t \, , \tau^m(x)>t) = \ind{\tau^{m,i}x) > t}E\Big[e^{-\int_t^T \lambda^m_s(x) ds}|\ccH_t\Big].
  \end{align}
  This assumption precisely reflects what we call a \emph{homogeneous} cohort of insured persons: while individual surrender or mortality occurs of course individually, the estimated probability distributions are the same for all individuals in the considered cohort.	\hfill $\diamond$
  \end{remark}

  In case of early surrender at time $T_i$, the policyholder can reinvest the surrender benefit ${\rm SB}(T_i)$ for the remaining time until the maturity $T$ leading to the benefit
  \begin{align}\label{eq:SB payoff}
 	   \SB(T_i) \beta(T_i,T)^{-1}
  \end{align}
  at maturity $T$. 
  
By $\tau^s$ we denote the time when the policy holder decides to exercise the surrender option. 
Similar to mortality, we assume that $\tau^s$ satisfied 
\begin{align} \label{def:taus}
      P(\tau^s > T \,|\,\ccH_t,\, \tau^s>t) = \ind{\tau^s > t}E\Big[e^{-\int_t^T \lambda^s_u ds}|\ccH_t\Big].
  \end{align}
  
  For further literature in this regard, and a detailed treatment of random times with intensities and associated enlargements of filtrations, see \cite{AksamitJeanblanc}. The approach chosen here falls in the class of so-called doubly-stochastic modeling approaches for the  random times and the related immersion hypothesis holds.
  Note, that even if the policy holder decides in continuous time, the exercise of the surrender option only takes place on a discrete grid: more precisely, if $T_{i-1}<\tau^s \le T_i$ then surrender takes place at $T_i$ leading to the payoff specified in Equation\eqref{eq:SB payoff}.

  \subsection{The valuation rule}
  The value of the variable annuity  decomposes in GMAB, SB and DB (all of course depending on the maturity $T$) in a linear way which allows to apply (for example) the QP-rule directly for IFA-free pricing. We denote the expectation under $P^*$ by $E^*=E_{P^*}$. We call the premium of an insurance-finance contract \emph{fair} if it does not introduce any IFA. In this sense, prices obtained by the QP-rule or by taking expectations under a measure $P^*$ as in Theorem \ref{thm-FTIFA} induce fair premia.

  \begin{proposition}\label{prop:VA}
  Assume that $P^*$ is as in Theorem \ref{thm-FTIFA}. Then, an IFA-free valuation of the variable annuity is given by 
  	\begin{align}\label{price: VA}
  		{\rm VA}_t = {\Pi}_t + {\rm GMAB}_t + {\rm SB}_t + {\rm DB}_t, \qquad 0 \le t \le T,
  	\end{align}
  	where \begin{enumeratei}
  		\item the \emph{discounted value of the future (and the present) payments} is 
  		\begin{align} \label{eq:Pi t}
{\Pi}_t = - \sum_{i=0}^n \beta(t,T_i) \pi_i \, \ind{T_i \ge t}E^*\Big[ \ind{\tau^m > T_i ,\, \tau^s> T_i} \mid \ccF_t \Big];  		
  		\end{align}
  		\item the value of the \emph{guaranteed minimum accumulation benefit} is 
  		\begin{align}\label{eq:GMAB2} 
{\rm GMAB}_t = \beta(t,T)\,  E^*\Big[ \ind{\tau^m(x) >T} \ind{\tau^s >T} \cdot \max (  F_T^\pi,  K_T^\pi )  \mid \ccF_t\Big];   			
  		\end{align}
  		\item the value of the \emph{death benefit}
  		\begin{align}\label{eq:DB}
 {\rm DB}_t = \sum_{k=1}^{n+1} \beta(t,T_k) \, E^*\bigg[  \ind{\tau^m(x) < \tau^s} \ind{T_{k-1}< \tau^m(x) \le T_k} \max(   F_{T_k}^\pi,  K_{T_k}^\pi )   \big| \ccF_t \bigg];
 \end{align}

  		\item the value of the \emph{surrender benefit} is
  		\begin{align} \label{eq:SB2} {\rm SB}_t=   \sum_{k=1}^n \beta(t,T_k) \, E^*\Big[  \ind{ T_{k-1}<\tau^s \le T_k}  \ind{\tau^s < \tau^m(x)} F_{T_k}^\pi p(T_k)   \mid \ccF_t \Big].\end{align}
  	\end{enumeratei}
  \end{proposition}
  
  \begin{proof}
  First, recall that the (non-random) payment $\pi_i$ is done at time $T_i$, $i = 0,\dots,n$. Since all quantities are deterministic, we sum the discounted payments after or at $t$, take conditional $P^*$-expectations and \eqref{eq:Pi t} follows. 
  
  Second, recall that the {\rm GMAB} provides a payoff  only  if the policyholder still alive until $T$ (i.e., $\{\tau^m(x) > T\}$) and if there was no surrender until this time  (i.e., $\{\tau^s > T\}$). The payoff $\max(F_{T_k}^\pi,K_{T_k}^\pi)$, see Equation \eqref{eq:GMAB}, is $\ccF_T$-measurable and paid at maturity $T$. Taking conditional $P^*$-expectations leads to the expression in \eqref{eq:GMAB2}. 
  
  Third, the death benefit provides a payoff only  if the surrender option  was not exercised before. Hence, if  $T_{k-1}<\tau^m(x) \le T_k$, the contract pays the amount $\max(   F_{T_k}^\pi,  K_{T_k}^\pi )$ at $T_k$; conditional on no previous surrender of course. Summing again over all possible time $T_1,\dots,T_{n+1}=T$ and taking $P^*$-expectations, Equation \eqref{eq:DB} follows.
  
 Fourth, we consider the surrender option. This option can only be exercised if the insurer is still alive (i.e., $\{\tau^s < \tau^m(x)\}$). 
 According to Equation \eqref{eq:SB}, the payoff at $T_k$ is $F^\pi_{T_k}p(T_k)$ if $T_{k-1} < \tau^s \le T_k$, i.e.\ if surrender is exercised between $T_{k-1}$ and $T_k$. Summing the relevant dates, discounting and taking $P^*$-expectations lead to Equation \eqref{eq:SB2}. 
  
 These four components add up to the price of the variable annuity. Hence the value in \eqref{price: VA} is an IFA-free valuation by Theorem \ref{thm-FTIFA}.
  \end{proof}

Valuing the fund $F^\pi$ itself is  possible by risk-neutral pricing when mortality and surrender are independent of $S$ \emph{and} sufficiently tractable. Otherwise, as we will show in the later sections, this task can become quite involved.

In the following, it  therefore remains to propose a tractable framework where the associated components of Equation \eqref{price: VA} can be computed efficiently. We proceed in two steps - first, we will introduce appropriate enlargement of filtration techniques and second, we will introduce an affine framework which in our eyes is flexible enough to capture all stylised facts in the market and on the other side allows tractable pricing rules.

\section{Progressive filtration enlargements}\label{sec: progressive}

To obtain a tractable pricing rule, we introduce an additional structure on the insurance filtration $\bbH$ by utilizing the theory of progressive enlargements.
For simplicity, we concentrate on the case where $\bbH$ is given by the smallest possibly enlargement of the publicly available information, given by the filtration $\bbF$, in our context: this is the enlargement with mortality and surrender times only\footnote{The filtration $\bbH$ can  additionally be enlarged by independent information without changing the approach. For example this could be mortality of additional clients, if those data does not contain additional information, which is the case when the mortality intensity is $\bbF$-measurable. The study of more complicated settings with incomplete information  remains a topic for future research.}.

First, we concentrate on a single  random time $\tau=\tau^m$, which  models the random life time of the insured client.  We consider multiple times in Section \ref{sec: multiple stopping times}.
The main assumption we make is to  assume that $\{\tau >t\}$, $t \ge 0$, are atoms in the enlarged filtration $\bbH$:  assume that for all $t \ge 0$,
\begin{align}\label{eq:progressive_atom}
	\ccH_t \medcap \{\tau > t \} = \ccF_t \medcap \{\tau > t\}.
\end{align}
This means that for each $H \in \ccH_t$ there exists an $F \in \ccF_t$, such that $H \cap \{\tau > t\} = F \cap \{\tau >t \}$ and vice versa.
A classical and well-known example where \eqref{eq:progressive_atom} holds is the progressive enlargement of $\FF$ with the random time $\tau$. %

\begin{remark}[Progressive enlargement] Assume that for all $t \ge 0$,
\begin{align}\label{eq progressive enlargement}
 	\ccH_t := \sigma\big( \{ F \medcap \{\tau \le s \}: F \in \ccF_t, s \le t\}\big) =: \ccF_t \vee \{\tau \wedge t \},
\end{align}
then $\bbH$ is called the \emph{progressive enlargement} of $\FF$ with $\tau$.
Then \eqref{eq:progressive_atom} holds: indeed, this follows from
$$ F \medcap \{ \tau>s \} \medcap \{ \tau > t \} = F \medcap \{ \tau > t \}, $$
for $0 \le s \le t \le T $ and $F \in \ccF_t. $
For a detailed study and many references to related literature see \cite{AksamitJeanblanc}.   Typically, we will be interested in a filtration which contains  more information, like  the employment status, health status, surrender behaviour, etc. 	
\hfill $\diamondsuit$
\end{remark}

For a pricing rule, fix $t$ and consider a measure $P^*$ as in Theorem \ref{thm-FTIFA}, for example $P^*=\QcirP$ as in the QP-rule. Then, fair insurance pricing works as follows:  the insurance premium $p_t$ is computed via taking expectations under $P^*$ with respect to $\ccF_t$. 
In this regard, denote by
\begin{align}
	G_t = P^* (\tau > t \,|\, \ccF_t), \qquad t \in  \TT,
\end{align}
the \emph{Az\'ema supermartingale} under $P^*$. The Az\'ema supermartingale  decodes the survival probability of $\tau$ in the smaller filtration $\FF$.

\begin{proposition}\label{prop EXtau}
Assume $\FF \subset \bbH$, $\tau$ being a $\bbH$-stopping time, and that \eqref{eq:progressive_atom} holds.
Then, for any $\bbF$-adapted process  $A$ bounded from below,
\begin{align}\label{temp 133}
	E^*[A_\tau  \,|\,\ccH_t] &= \ind{\tau \le t} A_\tau + \ind{\tau > t, \, G_t>0} G_t^{-1}E^*\Big[A_\tau \ind{\tau > t } \, |\,\ccF_t\Big]
\end{align}
holds for all $t \in \TT$.
\end{proposition}

\begin{proof}
	Consider $A_\tau = A_\tau \ind{\tau \le t} + A_\tau \ind{\tau >t}$. Since
    $ A_\tau \ind{\tau \le t}$ is $\ccH_t$-measurable, we obtain $E^*[A_\tau \ind{\tau \le t}|\ccH_t] = A_\tau \ind{\tau \le t} $,  the first addend of \eqref{temp 133}.

	Next, we consider $A_\tau \ind{\tau >t}$. Note that $E^*[A_\tau \ind{\tau >t} |\ccF_t]$ vanishes on the set $\{\tau > t, \, G_t=0\}$, and so does the second addend in  \eqref{temp 133}. So, we may assume $G_t>0$ everywhere. Now, \eqref{eq:progressive_atom} implies that for any $t \ge 0$ and any $\ccH_t$-measurable random variable $\tilde A_t$ we find a $\ccF_t$-measurable random variable $Y_t$, such that
	\begin{align}\label{eq: XY tau}
		\tilde A_t \ind{\tau > t} = Y_t \ind{\tau > t}
	\end{align}
	by an application of the monotone class theorem.  Using $E^*[A_\tau \ind{\tau > t}|\ccH_t]$ for $ \tilde A_t$ in this equation and taking conditional expectation with respect to $\ccF_t$, we obtain that
	\begin{align} 		
	G_t^{-1} E^*\Big[ E^*[A_\tau|\ccH_t] \ind{\tau > t} |\ccF_t \Big] &=
		G_t^{-1} E^*\Big[ Y_t \ind{\tau > t} |\ccF_t \Big] = Y_t. %
	\end{align}
On the other hand,
\begin{align*}
	E^*\Big[ E^*[A_\tau|\ccH_t] \ind{\tau > t} |\ccF_t \Big] &= E^*\Big[ A_\tau \ind{\tau > t} |\ccF_t \Big] 
\end{align*}
and the proof of the proposition is finished. 
\end{proof}

We add a simple corollary for the QP-rule which follows directly	by Equation \eqref{PQ-rule:filtation}. 

\begin{corollary}\label{cor:QP}
	Assume that $P^*=\QcirP$. Then,
	\begin{align}
		E^*\Big[A_\tau \ind{\tau > t } \, |\,\ccF_t\Big] = E_Q\Big[ E_P[A_\tau \ind{\tau > t }|\ccF_T] \, | \, \ccF_t\Big].
	\end{align}
\end{corollary}

\begin{example}[Doubly stochastic stopping times]
\label{ex:doubly stochastic}
The most important example regarding modeling with affine processes will be a doubly stochastic setting, which we introduce now. 
Consider an $\FF$-adapted, increasing process  $\Lambda=(\Lambda_t)_{t \in \TT}$ with $\Lambda_0=0$ together with an independent  exponential random variable $E$ and let\footnote{Such a construction is called \emph{doubly stochastic}. We refer to \cite{BieleckiRutkowski2002, Bremaud1981, GehmlichSchmidt2016MF} for details and references. If $\Lambda_t=\int_0^t \lambda_sds$, then $\lambda$ is called the \emph{intensity} of $\tau$. If the intensity is deterministic, it is called \emph{hazard rate}. }
	\begin{align} \label{def tau doubly stochastic}
	   \tau = \inf\{ t \in \TT: \Lambda_t \ge E \},
	\end{align}
	with the convention that $\inf \emptyset = T+1$.
Then, $\tau$ is not an $\FF$-stopping time, but it is a stopping time in the progressive enlargement $\bbH$. In particular, 	$$  
	G_t = P^*(\tau>t | \ccF_t)  = P^*( \Lambda_t < E |\ccF_t)  = e^{-\Lambda_t} $$
	for any $t \in \TT$.
	 Moreover,
	\begin{align}
		E^*[A_\tau \ind{\tau > t}| \ccH_t] &=
		\ind{\tau >t} G_t^{-1}
		\sum_{s=t+1}^T E^*[ A_s P^*(\tau = s \,|\,\ccF_s)|\ccF_t] \nonumber \\
		&= 
		\ind{\tau >t} e^{\Lambda_t} \sum_{s=t+1}^T E^*\big[A_s \big( e^{-\Lambda_{s-1}} - e^{-\Lambda_{s}} \big) |\ccF_t\big]. \label{eq affine Xtau}
	\end{align}
	Additionally, if $P^*=\QcirP$, then 
	$$ E^*\big[A_s \big( e^{-\Lambda_{s-1}} - e^{-\Lambda_{s}} \big) |\ccF_t\big] = 
	E_Q[A_s \big( e^{-\Lambda_{s-1}} - e^{-\Lambda_{s}} \big) |\ccF_t\big], $$
	 such that the evaluation of the payment stream $A$ (typically depending on financial quantities like interest rates and stock markets) can be done in a market-consistent way. 
	These formulae are key results for the  valuation of  a large number of hybrid products.
\hfill $\diamondsuit$
\end{example}
An additional difficulty arises when $A$ is not $\FF$-measurable. Here one is able to exploit the structure of the progressive enlargement in the case where $\tau$ is honest. For example, Corollary 5.12 in \cite{AksamitJeanblanc} allows to decompose $X$ in several $\FF$-adapted components on random intervals depending solely on $\tau$.

\begin{proposition}\label{prop EFXtau}
Consider $P^*=\QcirP$ and 
assume $\FF \subset \bbH$, $\tau$ being a $\bbH$-stopping time, and that \eqref{eq:progressive_atom} holds.
Then, for any $\bbH$-adapted process  $A$ and any $\ccF_T$-measurable random variable $F_T$, both bounded from below, it holds that for all $t \in \TT$,
\begin{align*}
	E_{\subQcirP}[F_T A_\tau  |\ccH_t] &= \ind{\tau \le t} A_\tau  E_{\subQcirP}[F_T|\ccH_t]\\
&+ \ind{\tau > t,\, G_t >0 }\; G_t^{-1}E_Q\Big[F_T E_P[A_\tau \ind{\tau > t} |\ccF_T] |\ccF_t\Big].
\end{align*}
\end{proposition}
The proof follows as for Proposition \ref{prop EXtau}.

\subsection{Multiple stopping times}\label{sec: multiple stopping times}
For the insurance company it is of course important to consider more than one stopping time, and to allow for dependence between these stopping times. We generalise the previous framework by considering a countable number of atoms and provide the associated generalisations of the previously obtained results. %

We start by providing the key result for multiple stopping times. To this end, 
consider  for each $t \in \TT$ $\ccH_t$-measurable sets   $(P_{t}^1,P_t^2,\dots,P_t^n)$  such that  $P_t^i \cap P_t^j = \emptyset$ for $i \neq j$.
Assume that 
	\begin{align}\label{eq:progressive_atom general}
		\ccH_t \medcap P_t^i = \ccF_t \medcap P_t^i,
		\qquad 1 \le i \le n, \ t \in \TT, 
	\end{align}
	and denote by \begin{align}\label{azema general}
		G_t^i := P^*(P_t^i \, | \,  \ccF_t), \quad t \in \TT,
	\end{align}
	the respective generalization of the Az\'ema supermartingale. Note that without additional assumptions $G^i$ does not need to be a supermartingale. Nevertheless, we have the following generalization of Proposition \ref{prop EXtau} directly formulated in terms of $\QcirP$. The respective version with general $P^*$ follows easily. Set $\Omega_t = \sum_{i=1}^n P_t^i$.

\begin{proposition}\label{prop EXtau general}
Assume $\FF \subseteq \bbH$, and that \eqref{eq:progressive_atom general} and \eqref{azema general} hold.
Then, for any $\ccH_T$-measurable  random variable   $A$ bounded from below,
\begin{align}\label{temp 133 g}
	 \Ind_{\Omega_t} E_{\subQcirP}[ A  \,|\, \ccH_t]  &= \sum_{i =1}^n
	\Ind_{P_t^i \cap \{G_t^i>0 \}} (G_t^i)^{-1}E_Q\Big[ E_P[A \, \Ind_{P_t^i} |\,\ccF_T] |\ccF_t\Big]
\end{align}
holds for all $t \in \TT$.
\end{proposition}

\begin{proof}
    First, we decompose
    \begin{align}\label{temp286}
    	\Ind_{\Omega_t} E_{\subQcirP}[A  | \ccH_t] = \sum_{i = 1} ^n  \Ind_{P_t^i} E_{\subQcirP}[A  \, \Ind_{P_t^i}| \ccH_t].
    \end{align}
	For the following, we fix $i$.
	From \eqref{eq:progressive_atom general}, a monotone class arguments gives the following:  for a $\ccH_t$-measurable random variable $X_t$ we can find a $\ccF_t$-measurable random variable $\tilde Y_t$, such that
	\begin{align}\label{eq:representation from atoms}
		X_t \Ind_{P_t^i} =  \tilde Y_t \Ind_{P_t^i}.
	\end{align}
	Hence, there exists an $\ccF_t$-measurable random variable $Y_t$ such that
	\begin{align}\label{temp291}
		 E_{\subQcirP}[A  | \ccH_t] \,  \Ind_{P_t^i}  = Y_t \Ind_{P_t^i}.
	\end{align}
 Taking conditional expectations with respect to $\ccF_t$ and multiplying with $\Ind_{P_t^i}$ yields that
	\begin{align*}
		\Ind_{P_t^i} E_{\subQcirP}\Big[ E_{\subQcirP}[A  \, \Ind_{P_t^i}| \ccH_t]| \ccF_t \Big]
		&= Y_t G_t^i \, \Ind_{P_t^i}.
	\end{align*}
	The left hand side equals $\Ind_{P_t^i} E_{\subQcirP}[ A  \, \Ind_{P_t^i} | \ccF_t ] $, which consequently vanishes when $G_t^i=0$. Inserting \eqref{temp291},  we obtain that
		\begin{align*}
		\Ind_{P_t^i \cap \{G_t^i>0\}}(G_t^i)^{-1} E_{\subQcirP}\big[ A \, \Ind_{P_t^i}| \ccF_t \big]
		&= E_{\subQcirP}\big[ A \,  \Ind_{P_t^i} | \ccH_t \big] \Ind_{P_t^i}.
	\end{align*}
	The claim now follows by combining this with \eqref{PQ-rule:filtation} and  \eqref{temp286}.
\end{proof}

\subsection{Two stopping times}\label{sec: two stopping times}
The most relevant case in this paper is the case with two stopping times, 
mortality $\tau^m$ and surrender $\tau^s$. 
In many papers it is assumed that they are independent or conditionally independent. This can be a serious restriction for the applications we have in mind: indeed, dependence between remaining life time and surrender is of course possible and should be taken into account. We may use the above result to do so.
Motivated by this, we will develop  Proposition \ref{prop EXtau general}  further in the case of two stopping times.

Consider two $\FF$-adapted, increasing processes $\Lambda^m$, and $\Lambda^s$ associated to $\tau^m$ and $\tau^s$, respectively. Assume that there exist two standard exponential random variables $E^m,$ and $E^s$, independent of $\ccF_T$ having continuous survival copula
	$$ \bar C(u_1,u_2):= P\big(\exp(-E^m)<u_1,\exp(-E^s)<u_2\big). $$
	Let
	$ \tau^m = \inf\{ t \in \TT: \Lambda_t^m \ge E^m\},$ and $\tau^s = \inf\{ t \in \TT: \Lambda_t^s \ge E^s\}  $; again we impose the convention that $\inf \emptyset = T+1$.
	Using independence of $E^m$, $E^s$ and $\ccF_T$, we obtain that
	\begin{align} \label{two stopping times:basic equation}
		P(\tau^m > t_1, \tau^s > t_2| \ccF_T ) &= P(\Lambda^m_{t_1} < E^m, \Lambda^s_{t_2} < E^s | \ccF_T )  \notag \\
		&= P\big(\exp(-E^m)<e^{-\Lambda^m_{t_1}},\ \exp(-E^s)<e^{-\Lambda^s_{t_2}} | \ccF_T\big) \notag\\
		&= \bar C\big(e^{-\Lambda^m_{t_1}},\ e^{-\Lambda^s_{t_2}}\big). 
	\end{align}
	
\begin{example}[Conditional independence] 
\label{ex:conditional independence}
If in addition, $E^m$ and $E^s$ are independent, then $\tau^m$ and $\tau^s$ are independent conditional on $\ccF_T$ and $C(u,v)=u \cdot v$. This simplifies the computations significantly since then 
\begin{align} \label{ex4.2}
P(\tau^m> t_1, \tau^s > t_2| \ccF_T )=\exp(-\Lambda^m_{t_1}-\Lambda^s_{t_2}).\end{align}
This is a highly tractable extension of the doubly stochastic setting from Example \ref{ex:doubly stochastic}.
	\hfill $\diamondsuit$
\end{example}

Fix $t \in \TT$ and consider the disjoint $\ccH_t$-measurable sets
\begin{align*}
	P_t^{1,1} &:= \{\tau^m >t, \tau^s >t\}, \\
	P_t^{2,u} &:= \{\tau^m >t, \tau^s = u\}, \ u = 0, \dots, t-1, \\
	P_t^{3,u} &:= \{\tau^m = u, \tau^s >t \},\ u=0,\dots,t-1, \\
	P_t^{4,u,v} &:= \{ \tau^m =u, \tau^s = v\},\ u,v=0,\dots,t-1,
\end{align*}
which will take the role of $P_t^1,\dots,P_t^n$ in the previous section.

Let 
\begin{align}
    \Gamma(t_1,t_2) :=&\  P(\tau^m> t_1, \tau^s > t_2| \ccF_T ) 
    =  P(\tau^m> t_1, \tau^s > t_2| \ccF_{\max \{t_1,t_2\}} ) \notag \\
    =&\     \bar C(\exp(-\Lambda^m_{t_1}),\exp(-\Lambda^s_{t_2})).
\end{align}

 Then we obtain from \eqref{two stopping times:basic equation} that
\begin{align}\label{G11}
\begin{aligned}	
	G_t^{1,1} &: = P(\tau^m > t, \tau^s > t| \ccF_T) =\Gamma(t,t), \\
	G_t^{2,u} &:= P(\tau^m > t, \tau^s  =u| \ccF_T ) = P(\tau^m > t, \tau^s  >u-1| \ccF_T ) - P(\tau^m > t, \tau^s  >u| \ccF_T ) \\ &\,\,=
        \Gamma(t,u-1)-\Gamma(t,u), \\
	G_t^{3,u} &:= P(\tau^m =u, \tau^s  >t| \ccF_T ) = \Gamma(u-1,t)-\Gamma(u,t), \\
	G_t^{4,u,v} &:= P(\tau^m =u, \tau^s  =v| \ccF_T )  \\
	&\,\,= \Gamma(u-1,v-1)-\Gamma(u,v-1)-\Gamma(u-1,v)+\Gamma(u,v),
	\end{aligned}
\end{align}
where again $u,v \in \{0,\dots,t-1\}$.

 Now, consider two $\FF$-adapted payment streams $A^m$ and $A^s$. If the insurer dies at $\tau^m$ before surrendering, he will receive $A^m_{\tau^m}$ while if he first surrenders, he will receive $A^s_{\tau^s}$. Precisely, this defines the following payoff:
 \begin{align}\label{eq: XtT two stopping times}
 	X_{t,T} =  \ind{\tau^m > t, \tau^s>t } ( \ind{\tau^m < T, \tau^m < \tau^s} A^m_{\tau^m} +  \ind{\tau^s < T, \tau^s \le  \tau^m} A^s_{\tau^s}).
 \end{align}

\begin{proposition}\label{prop 5.7}
If we consider $P^*=\QcirP$, then 
	for the payoff in \eqref{eq: XtT two stopping times},
	\begin{align}
		E_{\subQcirP}[ X_{t,T} | \ccH_t] &= \ind{\tau^m > t, \tau^s>t } (G_t^{1,1})^{-1} \sum_{u=t+1}^{T-1}  E_Q[A_u^m \,G_{u}^{3,u} + A_u^s \,G_{u-1}^{2,u} | \ccF_t]. \notag
	\end{align}
\end{proposition}

\begin{proof}
	For the  part with $A^m$ we obtain the following decomposition:
	\begin{align*}
			\ind{t<\tau^m< T, \ \tau^m < \tau^s} A^m_{\tau^m} %
			&= \sum_{u=t+1}^{T-1} \ind{\tau^m=u, \tau^s > u} A^m_u 
			= \sum_{u=t+1}^{T-1}\Ind_{P_{u}^{3,u}} \, A^m_u.
	\end{align*}
	Similarly,
	\begin{align*}
		\ind{t<\tau^s < T, \tau^m \ge  \tau^s} A^s_{\tau^s} &= \sum_{u=t+1}^{T-1} \ind{\tau^s=u, \tau^m > u-1} A^s_u
		= \sum_{u=t+1}^{T-1}\Ind_{P_{u-1}^{2,u}} A^s_u.
	\end{align*}
	The result now follows by applying Proposition \ref{prop EXtau general}.
\end{proof}

\section{An affine framework}\label{sec:affine}

Affine processes are a highly tractable class of processes and highly suited to the question at hand, in particular for modelling stochastic mortality term structures and surrender times which depend on the evolution of the stock market. We refer to  \cite{KellerResselSchmidtWardenga2019} for a detailed treatment of affine processes, including affine processes in discrete time.

In this section we propose a new framework in discrete time for the valuation of insurance products linked to financial markets. This generalizes existing approaches in three aspects: first, we incorporate the QP-approach in a general setting which allows for dependencies between financial markets, mortality, surrender and further factors. Note that this requires modelling the affine process under $P$ for the insurance quantities and under $Q$ for the parts of the affine process referring to the financial market. Second, it is important to acknowledge that the insurance part of the contract is monitored in discrete time, however. We therefore consider a discrete affine framework, i.e.~an affine process with stochastic discontinuities. In the existing works in the insurance literature, only stochastically continuous affine processes are used. Third,  we  introduce valuation formulas for more than one stopping time, to include stochastic mortality and surrender, for example. 

In this regard,  assume that there is a driving $d'$-dimensional process $Z=(Z_t)_{t \in \TT}$. The process $Z$ is  \emph{affine}, if it is a Markov process and its characteristic function is exponential affine. We additionally require the existence of exponential moments, such that
\begin{align}\label{eq:affineP}
	E[\exp(u Z_{t+1}) |Z_{t}] = \exp\big(A( u)+B( u) \cdot Z_t \big)
\end{align}
for $0 \le t < T$ and all $u \in \RR^{d'}$; with deterministic functions $A: \RR^{d'} \to \RR$ and $B: \RR^{d'} \to \RR^{d'}$. We will denote the state space of $Z$ by   $\cZ$. Typically  $\cZ=\RR^m_{\ge 0} \times \RR^n $ with $m+n=d'$.

In the light of the QP-rule it will be important to distinguish between the $\FF$-adapted parts of $Z$ and the parts which are only $\bbH$-adapted. This will become important when we consider the set of equivalent martingale measures $\ccM_e(\FF)$ where the choice of $\FF$ plays an important role. In this regard, let $Z=(X,Y)$ with $Z$ generating the publicly available information filtration $\FF$; here $X$ is $d_1$-dimensional and $Y$ is $d_2$-dimensional, with $d_1+d_2=d'$.   Note that $X$ might contain insurance-related quantities which are not publicly available.

We model the stock market as an exponential driven by the affine process with an additional drift, i.e.~we assume that discounted stock prices are given by
\begin{align}
\label{def:affine S}
	S_t = \exp( a_0  t +  a \cdot Y_t), \qquad t \ge 0,
\end{align}
with $(a_0,a) = (a_0,a_1,\dots,a_{d_2})\in \RR^{1+d_2}$. This modeling contains the Black-Scholes model and exponential L\'evy models as a special case (in discrete time).

Moreover, we assume that the conditional distribution of $X_t$, conditional on $Y_t$ and $X_{t-1}$ has affine form and denote
\begin{align}\label{eq:36}
	E_P[\exp(u \cdot X_t) | Y_t,X_{t-1}] &= \exp\Big( \alpha(u) + \beta(u) \cdot X_{t-1} + \gamma(u) \cdot Y_t\Big),
\end{align}
 for all $0 < t \le T$ and $u \in \RR^{d'}$. The coefficients $\alpha, \beta,$ and $\gamma$ can be computed from $A$ and $B$ from Equation \eqref{eq:affineP}.

To ensure absence of financial arbitrage, we assume that there exists an equivalent martingale measure $Q$. To achieve a high degree of  tractability, we assume that $Y$ is again affine under $Q$ (with existing exponential moments),
\begin{align}\label{eq:40}
	E_Q[\exp(u Y_{t+1}) |Y_{t}] = \exp\big(A_Q(u)+B_Q( u) \cdot Y_t \big)
\end{align}
 for all $0 \le t < T$ and $u \in \RR^{d_2}$.

\begin{remark}
 The Esscher change of measure is one well-known example which keeps the affine property during the measure change, but not the only one, see for example \cite{KallseMuhleKarbe2010}. %
\end{remark}

\subsection{The case of one stopping time}\label{sec:affine one stopping time}
For the modelling of a single random time $\tau$, we follow the doubly stochastic approach introduced in Example \ref{ex:doubly stochastic}. Precisely, we consider a non-decreasing process $\Lambda\ge 0 $ given by
$$ \Lambda_t = b_0 + \sum_{ s=0 } ^t \Big( b \cdot X_s + c \cdot Y_s \Big), \qquad t \ge 0$$
with $(b_0,b,c) \in \RR^{1+d'}$. Here, $b_0$ is chosen such that $\Lambda_0 = 0$,  guaranteeing in particular $P(\tau=0) = 0$, and $(b,c)$ are chosen such that $\Lambda$ is non-decreasing.

Denote by $\FF^Z$ the filtration  generated by $Z$.
We assume that $\tau$ is a doubly-stochastic stopping time, i.e.\ it satisfies \eqref{def tau doubly stochastic} with 
cumulative intensity $\Lambda$. 
Moreover, let  $\bbH$ be the progressive enlargement of $\FF^Z$ with $\tau$, $\ccH_t = \ccF^Z_t \vee (\tau \wedge t)$, for all $t \in \TT$.

For valuing the death benefit, we are interested in the claim $S_\tau \ind{t<\tau  \le T}$, which we decompose to
$$
	\sum_{s=t+1}^T S_s (\ind{\tau > s-1} - \ind{\tau >s}). $$
The following proposition  allows to value this claim as well as the claim $S_T \ind{\tau >T}$. 

Before this, we nee to introduce some notation.
We introduce the following recursive notation: define $\phi(T) = \alpha(-b) + A_Q( a -c + \gamma(-b))$,  $\psi^1(T) = \beta(-b)$, and  $\psi^2(T)=B_Q(a -c + \gamma(-b))$;
now let $u(s)=\psi^1(s+1)-b$, $v(s)=\psi^2(s+1)-c$ and set
\begin{align}\label{eq:recursion rule}
	\phi(s) & =\alpha(u(s)) + A_Q(v(s)+\gamma(u(s))), \notag \\ \
	\psi^1(s) & = \beta(u(s)), \\
	\psi^2(s) &= B_Q( v(s)+\gamma(u(s)) ) \notag
\end{align}
for $s=0,\dots,T-1$. Moreover,
$\phi'(T) =  A_Q( a )$,  $\psi'^1(T) =0$, and  $\psi'^2(T)=B_Q(a )$, following the same recursion rule \eqref{eq:recursion rule}.
Denote
\begin{align}
	\Phi(t,T)  = \sum_{s=t+1}^T \phi(s),
\end{align}
and, analogously, $\Phi'(t,T)= \sum_{s=t+1}^T \phi'(s)$. Moreover we write 
$$\psi(t+1) \cdot Z_t:= \psi^1(t+1) \cdot X_t + \psi^2(t+1) \cdot Y_t,$$ again analogously for $\psi'$.

\begin{proposition}\label{prop: affine computation}
	We have the following valuation-results
	\begin{align*}
		E_{\subQcirP} [S_T \ind{\tau > T} | \ccH_t] &= \ind{\tau > t}   e^{a_0 T  +   \Phi(t,T)  +\psi(t+1) \cdot Z_t  }, \\
	    E_{\subQcirP} [S_T \ind{\tau  > T-1} | \ccH_t] &= \ind{\tau > t}   e^{a_0 T   + \Phi'(t,T)  +\psi'(t+1) \cdot Z_t }.
	 \end{align*}
\end{proposition}
\begin{proof}
	Using Proposition \ref{prop EFXtau}, and the affine representations of $S$ and $\Lambda$, we obtain that
	\begin{align*}
		E_{\subQcirP} &[S_T \ind{\tau > T} | \ccH_t] = \ind{\tau > t} e^{\Lambda_t} E_{\subQcirP} [S_T e^{-\Lambda_T} | \ccF^Z_t] \\
		&= \ind{\tau > t}  E_{\subQcirP}\Big[ \exp\Big( a_0 T + a \cdot  Y_T - \sum_{s=t+1}^T \big( b \cdot X_s + c \cdot Y_s\big) \Big) \, \big| \,  Z_t \Big].		\end{align*}
	Now we proceed iteratively: first, consider the summation index $s=T$ in the above equation.
	Then, by Equation \eqref{eq:36} and by Equation \eqref{eq:40}, respectively, 
	\begin{align*}
		 E_P[ e^{ -b \,\cdot X_T } | \ccF^Z_{T-1} \vee Y_T ]
		     & = e^{\alpha(-b) + \beta(-b) \cdot X_{T-1} + \gamma(-b) \cdot Y_{T}}, \\
		 E_Q[ e^{ (a -c + \gamma(-b))  \cdot Y_T } | \ccF^Z_{T-1} ]
		     &=  e^{A_Q( a -c + \gamma(-b)) + B_Q(a -c + \gamma(-b)) \cdot Y_{T-1}}.
	\end{align*}
	Altogether, these two steps yield that 
	\begin{align}\label{affine QP rule}
		  E_{\subQcirP}[ e^{  u \cdot X_{s} + v \cdot Y_s  } | \ccF^Z_{s-1}]
		 &=  e^{\alpha(u) + \beta(u) \cdot X_{s-1}  + A_Q(v+\gamma(u)) + B_Q (v+\gamma(u)) \cdot Y_{s-1} }.
	\end{align}	
	With this formula, we can compute the next step (corresponding to the summation index $s=T-1$) where we choose 
	 $u=\beta(-b) -b$ and $v= B_Q(a-c+\gamma(-b)) - c$.
	Proceeding iteratively until $s=t+1$  yields
	 	\begin{align*}
		E_{\subQcirP} [S_T \ind{\tau > T} | \ccH_t] &= \ind{\tau > t}   e^{a_0 T   +\sum_{s=t+1}^T \phi(s) +\psi^1(t+1) \cdot  X_t + \psi^2(t+1) \cdot Y_t }, \end{align*}	
	the first claim.
	The second claim follows in a similar way. 	
	\end{proof}

\subsection{Survival and surrender} \label{sec:survival and surrender}
Since we are interested in modelling surrender and survival, we need to consider two stopping times. To exploit the full power of the affine framework, we will assume that they are conditionally independent, as in Example \ref{ex:conditional independence}. 
More general schemes could use copulas or 
could include self-exciting effects, like in \cite{ErraisGieseckeGoldberg2010} which, however, seems less important for the insurance application we have in mind. 

In this regard, let $\tau^m, \tau^s$ be conditionally independent, doubly stochastic random times as introduced in Section \ref{sec: two stopping times} with associated cumulated intensities $\Lambda^m$ and $\Lambda^s$. We assume that
\begin{align} 
	\label{def:Lambda i}
\Lambda_t^i = b_0^i + \sum_{s=0}^t\big(  b^i \cdot X_s + c^i \cdot Y_s\big), \qquad t \ge 0 
\end{align}
with $(b_0^i, b^i,c^i) \in \RR^{1+d'}$, $i\in \{m,s\}$. Again, the coefficients need to be chosen such that the processes start in $0$ and are increasing.

First, we obtain that under these assumptions %
by Equation \eqref{ex4.2} and Equations \eqref{G11} that
\begin{align} \label{eq:G^i affine}
	G_t^{1,1} & =  e^{-\Lambda^m_t - \Lambda^s_t} , \notag\\
	G_t^{2,u} &=   e^{-\Lambda^m_t} \Big( e^{ - \Lambda^s_{u-1}}  - e^{ - \Lambda^s_{u}} \Big), \\
	G_t^{3,u} &=   \Big( e^{ - \Lambda^m_{u-1}}  - e^{ - \Lambda^m_{u}} \Big) e^{-\Lambda^s_t} \notag
\end{align}
	with an analogous expression for $G^4$ (which  will not be used here).
	
From these expressions it is clear that we need to generalize our previous notions of $\phi$ and $\psi$.	
Fix $0 \le s < T$ and 
consider $\kappa=(\kappa^1,\kappa^2)
$, (we use $\kappa^1$ for the coefficients associated with $X$ and $\kappa^2$ for those associated with $Y$)  defined by  
\begin{align}
    \label{def:kappa}
	\kappa(t,s) =\kappa(a,t,s) := \begin{cases}
		(-b^m, a-c^m) & t=T, \\
		(-b^m  ,  - c^m ) & s <  t < T, \\
		(-b^m - b^s,  - c^m - c^s  ) & t   \le s,
	\end{cases} 
\end{align}
for $s \in \{0,\dots,T\}$. We highlight the dependence on $a$, whenever necessary, through the notation $\kappa(a,t,s)$.
Note that, with this definition $\kappa^1(T,T)=-b^m$ and $\kappa^2(T,T)=a-c^m$, for example.

Next, we define recursively, following \eqref{eq:recursion rule},
\begin{align}\label{eq:phi psi}
\begin{aligned}
	\phi(T,s) &= \alpha(\kappa^{1}(s,T)) + A_Q(\kappa^{2}(s,T)+\gamma(\kappa^{1}(s,T))), \\
	\psi^1(T,s) &=\beta(\kappa^{1}(s,T)), \\
	\psi^2(T,s) &= B_Q(\kappa^{2}(s,T)+\gamma(\kappa^{1}(s,T)))\\
	\phi(t,s) &= \alpha(u(t,s)) + A_Q(v(t,s)+\gamma(u(t,s))), \\
	\psi^1(t,s)  &= \beta(u(t,s)), \\
	\psi^2(t,s) &=B_Q( v(t,s)+\gamma(u(t,s)) ),
\end{aligned}\end{align}
with $u(t,s)=\psi^1(t+1,s)+ \kappa^{1}(t,s)$, $v(t,s)=\psi^2(t+1,s)+\kappa^{2}(t,s)$ and $t=0,\dots,T-1$.
Again, we write
\begin{align} \label{eq:Phi affine}
	\Phi(t,s,T) = \sum_{t'=t+1}^T \phi(t',s)
\end{align}
and $\psi^1 \cdot X + \psi^2 \cdot Y = \psi \cdot Z$. The coefficients $\phi',\ \psi'$ and $\Phi'$ are obtained by the same recursion, exchanging $b^m$ and $c^m$ for $b^s$ and $c^s$ in $\kappa$. 
With this notation,	we have the following valuation-results:
\begin{proposition}\label{prop: affine computation two stopping times}
 For $t \le s \le  T$, and on $\{\tau^m > t, \tau^s >t\}$,
	\begin{align*} 
    E_{\subQcirP} [S_T\ind{\tau^m > T, \tau^s >s} | \ccH_t]
	&=
		e^{\Lambda^m_t + \Lambda^s_t} E_{\subQcirP} [S_T e^{-\Lambda^m_T - \Lambda^s_s} | \ccF^Z_t] \\
	&=   e^{a_0 T   +\Phi(t,s,T) +\psi(t+1,s) \cdot  Z_t  } \\
     E_{\subQcirP} [S_T\ind{\tau^m > s, \tau^s >T} | \ccH_t] 
	    &=e^{\Lambda^m_t + \Lambda^s_t} E_{\subQcirP} [S_T e^{-\Lambda^m_s - \Lambda^s_T} | \ccF^Z_t] \\
	    &=   e^{a_0 T   + \Phi'(t,s,T) +\psi'(t+1,s) \cdot  Z_t  }.
	 \end{align*}
\end{proposition}

\begin{proof}
	We proceed iteratively as in the proof of Proposition \ref{prop: affine computation}. %
  Note that, by Equations \eqref{def:affine S} and \eqref{def:Lambda i},
	\begin{align} 
	\label{temp:919}
	S_T \cdot e^{\Lambda^m_t + \Lambda^s_t} \cdot   e^{-\Lambda^m_T - \Lambda^s_s} = e^{a_0 T + a \cdot Y_T -  \sum_{i=t+1}^T  (b^m X_i+c^m Y_i)  -  \sum_{j=t+1}^s  (b^s X_j + c^s Y_j)}. \end{align}
	For the time point  $T$ with $s<T$,  we obtain with \eqref{affine QP rule}
	\begin{align}
		 E_{\subQcirP}[ e^{ -b^m \cdot X_T +(a-c^m) \cdot Y_T } | \ccF^Z_{T-1}] &=
		   e^{\alpha(-b^m) + \beta(-b^m) \cdot X_{T-1} + A_Q(a-c^m+\gamma(-b^m)) + B_Q(a-c^m + \gamma(-b^m)) \cdot Y_{T-1}} \notag\\
		   & = 	 e^{\phi(T,T) + \psi^1(T,T) \cdot X_{T-1} + \psi^2(T,T) \cdot Y_{T-1} }. \label{eq:944}
	\end{align}
	Note that $\phi(T,T)=\phi(T,s)$ and $\psi(T,T)=\psi(T,s)$ by definition.
	For the time point $T-1$, we have to compute  (since $s \le T-1$)
		\begin{align}
		\lefteqn{
	 E_{\subQcirP} \Big[ e^{ (\psi^1(T,s) - b^m) \cdot X_T  + (\psi^2(T,s) -c^m) \cdot Y_T} | \ccF^Z_{T-1} \Big] } \qquad \qquad \notag\\
	 & = E_{\subQcirP} \Big[ e^{ u(T-1,s) \cdot X_T  + v(T-1,s) \cdot Y_T} | \ccF^Z_{T-1}\Big] \notag\\[2mm]
	 &= e^{\phi(T-1,s) + \psi^1(T-1,s) \cdot X_{T-1} + \psi^2(T-1,s) \cdot Y_{T-1}},
     \label{eq:954}
	\end{align}
where we again used Equation \eqref{affine QP rule}.
For time points $t'\le T-1$ we have to consider the two cases  $t'>s$ and $t' \le s$. For the first case, we obtain as above that 
		\begin{align*}
	 E_Q[ e^{ (\psi^1(t',s) - b^m) \cdot X_{t'}  + (\psi^2(t',s) -c^m) \cdot Y_{t'}} | \ccF^Z_{t'-1}]= e^{\phi(t'-1,s) + \psi^1(t'-1,s) \cdot X_{t'-1} + \psi^2(t'-1,s) \cdot Y_{t'-1}}.
	\end{align*}
	For $t' \le s$ on the other side, we have to compute, see Equation \eqref{temp:919},
	\begin{align*}
		 \lefteqn{E_{\subQcirP}\Big[ e^{ (\psi^1(t',s) - b^m-b^s) \cdot X_{t'}  + (\psi^2(t',s) -c^m-c^s) \cdot Y_{t'}} | \ccF^Z_{t'-1}\Big] } \qquad \qquad \\ 
		 &= E_{\subQcirP}\Big[ e^{ u(t',s) \cdot X_{t'}  + v(t',s) \cdot Y_{t'}} | \ccF^Z_{t'-1}\Big] \\[2mm]
		 &=e^{\phi(t'-1,s) + \psi^1(t'-1,s) \cdot X_{t'-1} + \psi^2(t'-1,s) \cdot Y_{t'-1}}
	\end{align*}
	and the first claim follows for $s<T$. For the second claim, note that we have to consider
	\begin{align} 
	\label{temp:919b}
	S_T \cdot e^{\Lambda^m_t + \Lambda^s_t} \cdot   e^{-\Lambda^m_s - \Lambda^s_T} = e^{a_0 T + a \cdot Y_T -  \sum_{i=t+1}^s  (b^m X_i+c^m Y_i)  -  \sum_{j=t+1}^T  (b^s X_j + c^s Y_j)}. \end{align}
	This means the recursion starts with coefficients $b^s$ and $c^s$ instead of $b^m$ and $c^m$ which gives the coefficients $\phi'$, and $\psi'$ and the second claim follows.

    Finally, if $s=T$, at the final time point $T$, we have to modify \label{eq:944} and consider
    \begin{align*}
    E_{\subQcirP}[ e^{ (-b^m-b^s) \cdot X_T +(a-c^m-c^s) \cdot Y_T } | \ccF^Z_{T-1}]
    \end{align*}
    instead. This is reflected by the choice of $\kappa(T,T)$ in Equation \eqref{def:kappa} and the conclusion follows.
\end{proof}

In the next step we extend this proposition to allow for fractions $\nicefrac{S_T}{S_{T'}}.$ While for $T' \le t$ this immediate, the following results gives the extension for $T' > t.$ Also the case $T'=T$ can be obtained readily from Proposition \ref{prop: affine computation two stopping times} by letting $a_0=a=0$.

Define recursively, as above
\begin{align}\label{eq:phi psi'}
\begin{aligned}
	\psi^2(T,s,T') &= B_Q(\kappa^{2}(s,T)+\gamma(\kappa^{1}(s,T)))\\
	\psi^2(t,s,T') &=B_Q( v(t,s)+\gamma(u(t,s)) ) - \ind{t-1=T'}a.
\end{aligned}\end{align}
and let
$\psi^1(t,s) \cdot X + \psi^2(t,s,T') \cdot Y =: \psi(t,s,T') \cdot Z$. The coefficient $\psi'$ is obtained by the same recursion, exchanging $b^m$ and $c^m$ for $b^s$ and $c^s$ in $\kappa$.

\begin{proposition}\label{prop: affine computation two stopping times with T'}
 For $t \le s \le T$ and  $t < T' < T$,
	\begin{align*} E_{\subQcirP} \Big[\frac{S_T}{S_{T'}}e^{-\Lambda^m_T - \Lambda^s_s} | \ccF^Z_t\Big]
	&=   e^{a_0 (T-T')   +\Phi(t,s,T) +\psi(t+1,s,T') \cdot  Z_t  } \\
	    E_{\subQcirP} \Big[\frac{S_T}{S_{T'}}e^{-\Lambda^m_s - \Lambda^s_T} | \ccH_t\Big] 
	    &=   e^{a_0 (T-T')   + \Phi'(t,s,T) +\psi'(t+1,s,T') \cdot  Z_t  }.
	 \end{align*}
\end{proposition}

\begin{proof}
The proof proceeds similarly as in the proof of Proposition \ref{prop: affine computation two stopping times}. To start with, we note that, by Equations \eqref{def:affine S} and \eqref{def:Lambda i},
	\begin{align} 
	\label{temp:919c}
	\frac{S_T}{S_{T'}} \cdot e^{\Lambda^m_t + \Lambda^s_t} \cdot   e^{-\Lambda^m_T - \Lambda^s_s} = e^{a_0 (T-T') + a \cdot (Y_T-Y_{T'}) -  \sum_{i=t+1}^T  (b^m X_i+c^m Y_i)  -  \sum_{j=t+1}^s  (b^s X_j + c^s Y_j)}. \end{align}
	For the time point  $T>T'$,  we recall from  \eqref{eq:944}
	\begin{align*}
		 E_{\subQcirP}[ e^{ -b^m \cdot X_T +(a-c^m) \cdot Y_T } | \ccF^Z_{T-1}] 
		   & = 	 e^{\phi(T,T) + \psi^1(T,T) \cdot X_{T-1} + \psi^2(T,T) \cdot Y_{T-1} }.
	\end{align*}
	Note that $\phi(T,T,T')=\phi(T,s,T')$ and $\psi(T,T,T')=\psi(T,s,T')$ by definition.
	For the time point $T-1$, we have to compute  (since $s \le T-1$)
		\begin{align*}
		\lefteqn{
	 E_{\subQcirP} \Big[ e^{ (\psi^1(T,s) - b^m) \cdot X_T  + (\psi^2(T,s) -c^m) \cdot Y_T- a Y_{T'}} | \ccF^Z_{T-1} \Big] } \qquad \qquad \\
	 &= e^{\phi(T-1,s) + \psi^1(T-1,s) \cdot X_{T-1} + \psi^2(T-1,s) \cdot Y_{T-1} - a Y_{T'} },
	\end{align*}
according to Equation \eqref{eq:954}.
It becomes apparent, that for $T-1=T'$, $\psi^2(T-1,s)$ from the recursion in Proposition \ref{prop: affine computation two stopping times} has to be replaced by $\psi^2(T-1,s)-a$ and the claim follows easily. 
\end{proof}

For the next result, we first extend \eqref{eq:40} to arbitrary times. For $t < T$, we define recursively
\begin{align}\label{eq: AQ and BQ}
\begin{aligned}
    A_Q(T-1,T) &= a_0 T + A_Q(a), \\
    A_Q(t,T) &= A_Q(t+1,T) +  A_Q(B_Q(t+1,T)), \\
    B_Q(T-1,T) &= B_Q(a) \\ 
    B_Q(t,T) &= B_Q(B_Q(t+1,T)).
    \end{aligned}
\end{align}
We will also use the notation $A_Q(t,T,a_0,a)$ and $B_Q(t,T,a)$ to highlight the dependence on $a_0$ and $a$, whenever necessary.

\begin{proposition}
    Under Equation \eqref{eq:40}, we have that
    \begin{align}\label{eq:affine property Y II}
        E_Q[S_T|\ccF_t^Z ]= e^{A_Q(t,T) + B_Q(t,T) \cdot Y_t}
    \end{align}
    with the coefficients $A_Q(t,T)$ and $B_Q(t,T)$ defined in Equation \eqref{eq: AQ and BQ}.
\end{proposition}
\begin{proof}
    This result follows by an iterated application of \eqref{eq:40}. Indeed, we have that
    \begin{align*}
        E_Q[S_T|\ccF_{T-1}^Z ]
        &= e^{a_0 T + A_Q(a) + B_Q(a) \cdot Y_{T-1}} =
        e^{A_Q(T-1,T) + B_Q(T-1,T) \cdot Y_{T-1}}.
    \end{align*}
    Moreover, for any $s \ge t$ and $s \le T$,
    \begin{align*}
        E_Q[e^{A_Q(s+1,T) + B_Q(s+1,T) \cdot Y_{s+1}}|\ccF_{s}^Z ]
        &= e^{A_Q(s+1,T) +  A_Q(B_Q(s+1,T)) + B_Q(B_Q(s+1,T)) \cdot Y_{s}}
    \end{align*}
    and the claim follows.
\end{proof}

For the following result, we need to handle the dependence on $a$ specifically depending on time. 
We therefore define recursively, as above
\begin{align}\label{eq:phi psi' II}
\begin{aligned}
	\psi^2(T,s,T',a,\bar a) &= B_Q(\kappa^{2}(a,s,T)+\gamma(\kappa^{1}(a,s,T)))\\
	\psi^2(t,s,T',a,\bar a) &=B_Q( v(a,t,s)+\gamma(u(a,t,s)) ) - \ind{t-1=T'}\bar a.
\end{aligned}\end{align}
with $u(a,t,s)=\psi^1(t+1,s)+ \kappa^{1}(a,t,s)$, $v(a,t,s)=\psi^2(a,t+1,s)+\kappa^{2}(a,t,s)$
and let
$\psi^1(t,s) \cdot X + \psi^2(t,s,T',a,\bar a) \cdot Y =: \psi(t,s,T',a,\bar a) \cdot Z$. The coefficient $\psi'$ is obtained by the same recursion, exchanging $b^m$ and $c^m$ for $b^s$ and $c^s$ in $\kappa$.

\begin{proposition}\label{prop: affine computation two stopping times with T' II}
 For $t \le s < T$ and  $t < T' < s$,
	\begin{align*} E_{\subQcirP} \Big[\frac{S_T}{S_{T'}}e^{-\Lambda^m_s - \Lambda^s_s} | \ccF^Z_t\Big]
	&=   e^{A_Q(s,T)-a_0 T+ \Phi(t,s,s)+\psi(t+1,s,T',B_Q(s,T),a) \cdot Z_t}.
	 \end{align*}
     On the other side, if $s \le T' \le T$, then
     \begin{align*} E_{\subQcirP} \Big[\frac{S_T}{S_{T'}}e^{-\Lambda^m_s - \Lambda^s_s} | \ccF^Z_t\Big]
        &=e^{-a_0T' + A_Q(T',T) +A_Q(s,T',0,(B_Q(T',T)-a))} \\
        &\quad \cdot
        e^{\Phi\big(t,s,s,B_Q(s,T',0,(B_Q(T',T)-a))\big) + 
        \psi\big(t+1,s,B_Q(s,T',0,(B_Q(T',T)-a))\big)\cdot Z_t} \\
        & \quad \cdot e^{-\Lambda_t^m - \Lambda_t^s}.
    \end{align*}
\end{proposition}

\begin{proof}
    We use iterated conditional expectations and the affine property to prove this result. Note that, in the case where $T' < s$, by Equation \eqref{eq:affine property Y II},
    \begin{align*}
        E_{\subQcirP} \Big[\frac{S_T}{S_{T'}}e^{-\Lambda^m_s - \Lambda^s_s} | \ccF^Z_t\Big] &=
        E_{\subQcirP} \Big[\frac{e^{A_Q(s,T)+B_Q(s,T)\cdot Y_s}}{S_{T'}}e^{-\Lambda^m_s - \Lambda^s_s} | \ccF^Z_t\Big]
         \\
         &=e^{A_Q(s,T)-a_0 T'} E_{\subQcirP} \Big[\frac{e^{B_Q(s,T)\cdot Y_s}}{e^{a \cdot Y_{T'}}}
         e^{-\Lambda^m_s - \Lambda^s_s} | \ccF^Z_t\Big].
    \end{align*}
    With a view on Equation \eqref{temp:919c}, we can now apply Proposition \ref{prop: affine computation two stopping times with T'}, however replacing $a \cdot Y_s$ by $B_Q(s,T) \cdot Y_s$. This is captured by adjusting the coefficients accordingly, using the notation from \eqref{eq:phi psi' II} and we obtain that
    \begin{align*}
        E_{\subQcirP} \Big[\frac{S_T}{S_{T'}}e^{-\Lambda^m_s - \Lambda^s_s} | \ccF^Z_t\Big] &=
        e^{A_Q(s,T)-a_0 T+ \Phi(t,s,s)+\psi(t+1,s,T',B_Q(s,T),a) \cdot Z_t}.
    \end{align*}
    On the other side, if $T' \ge s$,
    \begin{align*}
        E_{\subQcirP} \Big[\frac{S_T}{S_{T'}} | \ccF_s^Z\Big] 
        &=E_{\subQcirP} \Big[e^{-a_0T' + A_Q(T',T) + (B_Q(T',T)-a) \cdot Y_{T'} }  | \ccF_s^Z\Big]  \\
        &= e^{-a_0T' + A_Q(T',T)} \cdot {e^{A_Q(s,T',0,(B_Q(T',T)-a)) + B_Q(s,T',0,(B_Q(T',T)-a))\cdot Z_s}}.
    \end{align*}
    Using Proposition \ref{prop: affine computation two stopping times}, we obtain 
    \begin{align*}
        \lefteqn{E_{\subQcirP} \Big[\frac{S_T}{S_{T'}}e^{-\Lambda^m_s - \Lambda^s_s} | \ccF^Z_t\Big] } \qquad \\
        &=
        e^{-a_0T' + A_Q(T',T) +A_Q(s,T',0,(B_Q(T',T)-a))} \cdot
        E_{\subQcirP} \Big[e^{B_Q(s,T',0,(B_Q(T',T)-a))\cdot Z_s} e^{-\Lambda^m_s - \Lambda^s_s} | \ccF_t^Z \Big] \\
        &=e^{-a_0T' + A_Q(T',T) +A_Q(s,T',0,(B_Q(T',T)-a))} \\
        &\quad \cdot
        e^{\Phi\big(t,s,s,B_Q(s,T',0,(B_Q(T',T)-a))\big) + 
        \psi\big(t+1,s,B_Q(s,T',0,(B_Q(T',T)-a))\big)\cdot Z_t} \\
        & \quad \cdot e^{-\Lambda_t^m - \Lambda_t^s}
    \end{align*}
    and the proof is finished.
\end{proof}

\section{The valuation of the variable annuity}
\label{sec:valuation}

Coming back to the IFA-free valuation of a variable annuity, we now can utilize the obtained results on affine processes.

In the following, we utilise $\Phi$, $\psi$, $\Phi'$ and $\psi'$ as defined in Equations \eqref{eq:phi psi},   \eqref{eq:Phi affine} and \eqref{eq:phi psi'}. 
The special case of just computing the discounted value of future payments without reference to the stock price can be computed by letting $a_0=0$ and $a=0$. To emphasize this, we will write 
\begin{align} \label{eq:notation Phi}
\Phi^0(t,u,T) 
\end{align}
for the function $\Phi(t,u,T)$ obtained from  \eqref{eq:phi psi},   \eqref{eq:Phi affine} and \eqref{eq:phi psi'} with $a=0$. We use a similar notation for  the function $\psi$. It will also be useful when we can modify the parameter $a$. For this, we will use the notation
\begin{align} \label{eq:notation Phi2}
\Phi(t,u,T,a) 
\end{align}
and, similarly for $\psi$.
For clarity, we summarise the assumptions made until now in the following assumption.

\begin{assumption}\label{ass:61}
    We assume that 
    \begin{enumerate}[(i)]
        \item The pricing measure $P^*=\QcirP$ is the QP-measure
        \item $Z=(X,Y)$ is the affine process which satisfies Equations \eqref{eq:36} and \eqref{eq:40} with parameters $\alpha, \beta, \gamma$ and $A_Q,B_Q$, respectively.
        \item The discounted stock price $S$ satisfies Equation \eqref{def:affine S} with parameters $a_0,a$.
        \item $\tau^m$ and $\tau^S$ are double stochastic random times with associated cumulated intensities $\Lambda^m$ and $\Lambda^s$ satisfying Equation \eqref{def:Lambda i} with parameters $b_0^i, b^i, c^i$ for $i \in \{m,s\}$, respectively.
        \item The functions $\psi$, $\Phi'$ and $\psi'$ are defined in Equations \eqref{eq:phi psi} and \eqref{eq:Phi affine} where we also use the notation defined in Equation \eqref{eq:notation Phi} and Equation \eqref{eq:notation Phi2}.
    \end{enumerate}
\end{assumption}

Recall that the contract details of the variable annuity were introduced in Section \ref{sec:contract details}.
\begin{proposition}
Assume that Assumption \ref{ass:61} holds.  Then, the IFA-free price of the surrender benefit is given by
	\begin{align*}
		 \SB_0 &= 
         \sum_{k=1}^n \beta(T_k)p(T_k) \cdot \sum_{i=1}^k \pi_i \sum_{t=T_{k-1}+1}^{T_k}
         e^{a_0(T_k-T_i)} \Big( e^{\Phi(0,t-1,t) + \psi(1,t-1,T_i) \cdot Z_0} - e^{\Phi(0,t,t) + \psi(1,t,T_i) \cdot Z_0} \Big). 
	\end{align*}
\end{proposition}
\begin{proof}
	By Proposition \ref{prop:VA}, the surrender benefit is given by 
	$$ {\rm SB}_0=   \sum_{k=1}^n \beta(T_k)p(T_k) \, E_{\subQcirP}\Big[  \ind{ T_{k-1}<\tau^s \le T_k}  \ind{\tau^s < \tau^m} F^\pi_{T_k}   
	 \Big].$$
The value of the fund $F^\pi$ can be deduced from Equation \eqref{eq:F pi}, such that
we have to compute
\begin{align}
    E_{\subQcirP}\Big[  \ind{ T_{k-1}<\tau^s \le T_k}  \ind{\tau^s < \tau^m} \frac{S_{T_k}}{S_{T_i}}   
	 \Big]
\end{align}
for all $i \in \{ 1, \dots,k\}$. We write short $X^i_{T_k}:= \nicefrac{S_{T_k}}{S_{T_i}} $.

As in Proposition \ref{prop 5.7} we obtain that 
\begin{align*}
	E_{\subQcirP}\Big[  \ind{ T_{k-1}<\tau^s \le T_k}  \ind{\tau^s < \tau^m} X^i_{T_k}   
	 \Big] &= \sum_{t=T_{k-1}+1}^{T_k}
	 E_{\subQcirP}\Big[  \ind{ \tau^s =t }  \ind{t < \tau^m} X^i_{T_k}   
	 \Big] \\
	 &= \sum_{t=T_{k-1}+1}^{T_k}
	 E_{\subQcirP}\Big[  E_P\big[\ind{ \tau^s =t }  \ind{t < \tau^m} |\ccF^Z_{T_K}\big] X^i_{T_k}   
	 \Big] \\
	 &= \sum_{t=T_{k-1}+1}^{T_k}
	 E_{\subQcirP}\Big[  E_P\big[\ind{ \tau^s =t }  \ind{t < \tau^m} |\ccF^Z_t\big] X^i_{T_k}   
	 \Big]  \\
	 &= \sum_{t=T_{k-1}+1}^{T_k}
	 E_{\subQcirP}\Big[  G_t^{3,t} \, X^i_{T_k}   
	 \Big]. 
\end{align*}
By Equation \eqref{eq:G^i affine}, we obtain that
	\begin{align*}
		G^{3,t}_t = \Big( e^{ - \Lambda^s_{t-1}}  - e^{ - \Lambda^s_{t}} \Big) e^{-\Lambda^m_t}.
	\end{align*}
	Hence, from Proposition \ref{prop: affine computation two stopping times with T'}, 
	\begin{align*}
		E_{\subQcirP}[G^{3,t}_tX^i_{T_k}] &=  E_{Q}\Big[ \frac{S_{T_k}}{S_{T_i}} \Big( e^{ - \Lambda^s_{t-1}}  - e^{ - \Lambda^s_{t}} \Big) e^{-\Lambda^m_t} \Big] \\
		&= e^{a_0(T_k-T_i)} \Big( e^{\Phi(0,t-1,t) + \psi(1,t-1,T_i) \cdot Z_0} - e^{\Phi(0,t,t) + \psi(1,t,T_i) \cdot Z_0} \Big). \qedhere
	\end{align*}

\end{proof}

\bigskip 

For the following it is important to notice that $F^\pi$ is no longer affine, even if $S$ is. The computation of tractable formulas therefore becomes more complicated. If there is only a single investment, i.e.\ when $n=1$, then the following result provides the associated valuation formulas.  If the fund $F^\pi$ can be well approximated by an affine process, the  result also suffices of course. For the general case one needs to rely on more complicated Fourier decompositions, see for example Lemma 10.3 in \cite{Filipovic2009}.

Define 
\begin{align}
	\widetilde f(w,\lambda, \bar A) = \frac 1 {2 \pi} e^{(w + i \lambda) \bar A} \frac{ K^{-(w-1+i\lambda)}}{(w+ i \lambda)(w-1+i \lambda)}
\end{align}
for any $w>1$.

\begin{proposition}\label{prop:62}
Assume that Assumption \ref{ass:61} holds and in addition that $n=1$.  Then, the IFA-free valuation of the GMAB is as follows:
	\begin{align*}
		{\rm GMAB}_0   &= P_T^\pi \cdot  e^{\Phi^0(0,T,T)+ \psi^0(1,T) \cdot Z_0} \nonumber  \\
		&+ \frac 1 {2 \pi} \int e^{\Phi(0,T,T,(w+i\lambda)a) + \psi(1,T,T_1,(w+i\lambda)a) \cdot Z_0} \, \widetilde f(w,\lambda,a_0(T-T_1)) d \lambda.
	\end{align*}
\end{proposition}
\begin{proof}
	 To begin with, we observe that for $n=1$ and $t \ge  T_1,$
    \begin{align}
        F^\pi_t = \pi_1\frac{S_t}{S_{T_1}}.
    \end{align}
Then, following Proposition \ref{prop:VA}, with $K=\nicefrac{K'}{\pi_1}$,
	\begin{align}
		E_{\subQcirP}\Big[\ind{\tau^m >T, \tau^s >T} \max (  F_T^\pi, K') \Big] 
		&= K' E_{\subQcirP}\Big[e^{- \Lambda^m_T - \Lambda^s_T}\Big] + 
		\pi_1 E_{\subQcirP}\Big[ \Big(  \frac{S_T}{S_{T_1}} - K\Big) ^+ e^{- \Lambda^m_T - \Lambda^s_T} \Big],
	\end{align}
	such that we have to value a call on $\nicefrac{S_T}{S_{T_1}}$ with strike $K$. This can be done in affine models by relying on Fourier inversion, and we follow Chapter 10 in \cite{Filipovic2009}. First, note that, using \eqref{eq:notation Phi},
	\begin{align}\label{temp:1301}
		E_{\subQcirP}\Big[e^{- \Lambda^m_T - \Lambda^s_T}\Big] &= e^{\Phi^0(0,T,T)+ \psi^0(1,T) \cdot Z_0}. 
	\end{align}
	In the following, we will condition on $\ccF^Z_{T_1}$ and use the affine property. 
	Next, observe that
	\begin{align*}
		\frac{S_T}{S_{T_1}} 
		& = e^{a_0 (T-T_1) + a \cdot Y_T -a \cdot Y_{T_1}}.	
		\end{align*}
	To obtain an affine representation we stack the vector $\bar Z = (Z_T,Z_{T_1})$ and denote the associated coefficients by
	\begin{align}
		\bar B := \big((0,a),(0,-a)\big)
	\end{align}
	together with $\bar A = a_0(T-T_1)$
	such that 
	\begin{align*}
	\frac{S_T}{S_{T_1}}  = e^{\bar A + \bar B \cdot \bar Z}. 	
	\end{align*}
	As in Corollary 10.4 in \cite{Filipovic2009} we have the integral representation
	$$
	\Big( e^{\bar A + \bar B \cdot  z} - K \Big)^+ = \frac 1 {2 \pi} \int e^{ (w+i\lambda) \bar B \cdot  z} \, \widetilde f(w,\lambda, \bar A) d \lambda
	$$
	for some appropriate $w >1$.
	An application of Fubini's theorem yields
	\begin{align*}
		\lefteqn{ E_{\subQcirP}\Big[ \Big(  \frac{S_T}{S_{T_1}} - K\Big) ^+ e^{- \Lambda^m_T - \Lambda^s_T} \Big] } \qquad \qquad \\
		&=  E_{\subQcirP} \Big[ \Big(e^{\bar A + \bar B \cdot \bar Z} - K\Big)^+ e^{- \Lambda^m_T - \Lambda^s_T} \Big]  \\
		&= \frac 1 {2 \pi} \int E_{\subQcirP} \Big[ e^{(w+i\lambda) \bar B \cdot  \bar Z} e^{- \Lambda^m_T - \Lambda^s_T} \Big] \, \widetilde f(w,\lambda, \bar A) d \lambda. 
	\end{align*}
	The inner expectation computes to
	\begin{align}
		E_{\subQcirP} \Big[ e^{(w+i\lambda) \bar B \cdot  \bar Z} e^{- \Lambda^m_T - \Lambda^s_T} \Big] &= E_{\subQcirP} \bigg[\frac{e^{(w+i\lambda)a \cdot Y_T}}{e^{(w+i\lambda)a \cdot Y_{T_1}}} e^{- \Lambda^m_T - \Lambda^s_T}\bigg]. 
		\label{temp 1198}
	\end{align}
	This can be computed using the first expression in Proposition \ref{prop: affine computation two stopping times with T'} by choosing $a_0=0$ and letting $a$ be equal to $(w+i\lambda)a$. Using the notation defined in Equation \eqref{eq:notation Phi2}, we obtain that
	\begin{align}\label{temp1209}
		\eqref{temp 1198} &= \exp\Big(\Phi(0,T,T,(w+i\lambda)a) + \psi(1,T,T_1,(w+i\lambda)a) \cdot Z_0\Big)
	\end{align} 
	and the  claim follows.
\end{proof}

The death benefit is more complicated and we introduce the following notation. Denote
\begin{align}
    {\rm DB}^1(t) &= P_{T_1}^\pi \Big( e^{-\Lambda^m_{t-1}} -e^{-\Lambda^m_t}\Big) e^{-\Lambda_t^s}.
\end{align}
Recall the notation for $A_Q$ and $B_Q$ in Equation \eqref{eq: AQ and BQ} and let for  $t \ge T_1$
\begin{align}
    {\rm DB}^2(t) &= 
e^{A_Q(t,T,0,(w+i\lambda)a)+ \Phi(0,t,t,0,(w+i\lambda)a)+\psi(1,t,T_1,B_Q(t,T,(w+i\lambda)a),(w+i\lambda)a) \cdot Z_0}
\end{align}
while for $t < T_1$, 
\begin{align}
    {\rm DB}^2(t) &= e^{ A_Q(T',T,0,(w+i\lambda)a) +A_Q(t,T_1,0,(B_Q(T_1,T,(w+i\lambda)a)-(w+i\lambda)a))} \notag\\
        &\quad \cdot
        e^{\Phi\big(0,t,t,B_Q(t,T_1,(B_Q(T_1,T,0,(w+i\lambda)a)-(w+i\lambda)a))\big) } \notag\\
        & \quad \cdot e^{ 
        \psi\big(1,t,B_Q(t,T_1,(w+i\lambda)a,(B_Q(T_1,T,(w+i\lambda)a)-(w+i\lambda)a))\big)\cdot Z_0}.
\end{align}
\begin{proposition}
Assume that Assumption \ref{ass:61} holds and in addition that $n=1$.  Then, the IFA-free valuation of the death benefit is as follows:
	\begin{align*}
		{\rm DB}_0   &= \beta(T_1) \cdot \sum_{t=1}^{T_k} {\rm DB}^1(t) 
		+  \beta(T_1) \cdot \pi_1 \sum_{t=1}^{T_1} {\rm DB}^2(t) .
	\end{align*}
\end{proposition}
	
\begin{proof}	
	According to Proposition \ref{prop:VA}, with $n=1$,
	\begin{align*}
        {\rm DB}_0 =  \beta(T_1)  E_{\subQcirP}\Big[ \ind{\tau^m \in (0,T_1]} \ind{\tau^s > \tau^m} \max(F_{T_1}^\pi,K_{T_1}^\pi) \Big].
	\end{align*}
    We note that
    \begin{align*}
    E_{\subQcirP}\Big[ \ind{\tau^m \le T_1} \ind{\tau^s > \tau^m} \max(F_{T_1}^\pi,K_{T_1}^\pi) \Big] &= 
    \sum_{t=1}^{T_1}E_{\subQcirP}\Big[ \ind{\tau^m =t} \ind{\tau ^s > t} \max(F_{T_1}^\pi,K_{T_1}^\pi) \Big] \\
    &= \sum_{t=1}^{T_1}E_{\subQcirP}\Big[ \ind{\tau^m =t} \ind{\tau ^s >ts} \max(F_{T_1}^\pi,K_{T_1}^\pi) \Big] \\
    &=  \sum_{t=1}^{T_1}E_{\subQcirP}\Big[ G_t^{3,t}  \max(F_{T_1}^\pi,T_{T_1}^\pi) \Big]
    \end{align*}
	We therefore focus on 
	\begin{align*}
	\lefteqn{E_{\subQcirP}\Big[ G_t^{3,t}  \max(F_{T_1}^\pi,K') \Big]} \qquad\\
    &= K' E_{\subQcirP}\Big[ G_t^{3,t}  \Big]  + 
    \pi_1 \, E_{\subQcirP}\Big[ \Big( \frac{S_T}{S_{T_1}} - K\Big)^+G_t^{3,t}  \Big]. 
	\end{align*}
    Recall from Equation \eqref{eq:G^i affine} that
    \begin{align*}
        G_t^{3,t} &= \Big( e^{-\Lambda^m_{t-1}} -e^{-\Lambda^m_t}\Big) e^{-\Lambda_t^s}
    \end{align*}
	Proceeding as for Equation \eqref{temp:1301}, we obtain that
    \begin{align*}
        E_{\subQcirP}\Big[ G_t^{3,t}  \Big] &= \Big( e^{\Phi^0(0,t-1,t) + \psi^0(1,t-1) \cdot Z_0} - e^{\Phi^0(0,t,t) + \psi^0(1,t) \cdot Z_0} \Big).
    \end{align*}
    Moreover, 
    \begin{align*}
            \lefteqn{E_{\subQcirP}\Big[ \Big( \frac{S_T}{S_{T_1}} - K\Big)^+G_t^{3,t}  \Big] } \qquad \\
            &= E_{\subQcirP}\Big[ \Big( \frac{S_T}{S_{T_1}} - K\Big)^+ e^{-\Lambda^m_{t-1}-\Lambda_t^s}   \Big]
            - E_{\subQcirP}\Big[ \Big( \frac{S_T}{S_{T_1}} - K\Big)^+ e^{-\Lambda^m_{t}-\Lambda_t^s}   \Big].
    \end{align*}
    Now we obtain, as in the proof of Proposition \ref{prop:62},
    \begin{align*}
		\lefteqn{ E_{\subQcirP}\Big[ \Big(  \frac{S_T}{S_{T_1}} - K\Big) ^+ e^{- \Lambda^m_t - \Lambda^s_t} \Big] } \qquad \qquad \\
		&=  E_{\subQcirP} \Big[ \Big(e^{\bar A + \bar B \cdot \bar Z} - K\Big)^+ e^{- \Lambda^m_t - \Lambda^s_t} \Big]  \\
		&= \frac 1 {2 \pi} \int E_{\subQcirP} \Big[ e^{(w+i\lambda) \bar B \cdot  \bar Z} e^{- \Lambda^m_t - \Lambda^s_t} \Big] \, \widetilde f(w,\lambda, \bar A) d \lambda. 
	\end{align*}
	The inner expectation computes to
	\begin{align}
		E_{\subQcirP} \Big[ e^{(w+i\lambda) \bar B \cdot  \bar Z} e^{- \Lambda^m_t - \Lambda^s_t} \Big] 
        &= E_{\subQcirP} \bigg[\frac{e^{(w+i\lambda)a \cdot Y_T}}{e^{(w+i\lambda)a \cdot Y_{T_1}}} e^{- \Lambda^m_t - \Lambda^s_t}\bigg]. 
		\label{temp 1295}
	\end{align}
    For this, we can proceed exactly as in Proposition \ref{prop: affine computation two stopping times with T' II}, however replacing $a$ by $(w+i\lambda)a$. 
    This gives us, when $t \ge T_1$,
    \begin{align*}
        \eqref{temp 1295} &= 
        e^{A_Q(t,T,0,(w+i\lambda)a)+ \Phi(0,t,t,0,(w+i\lambda)a)+\psi(1,t,T_1,B_Q(t,T,(w+i\lambda)a),(w+i\lambda)a) \cdot Z_0}.
    \end{align*}
    On the other side, for $t < T_1$, we obtain that
     \begin{align*}
        &\eqref{temp 1295} = e^{ A_Q(T_1,T,0,(w+i\lambda)a) +A_Q(t,T_1,0,(B_Q(T_1,T,(w+i\lambda)a)-(w+i\lambda)a))} \\
        & \cdot
        e^{\Phi\big(0,t,t,B_Q(t,T_1,(B_Q(T_1,T,0,(w+i\lambda)a)-(w+i\lambda)a))\big) + 
        \psi\big(1,t,B_Q(t,T_1,(w+i\lambda)a,(B_Q(T_1,T,(w+i\lambda)a)-(w+i\lambda)a))\big)\cdot Z_0} %
     \end{align*}
	and the proof is finished.
\end{proof}

 \section{Conclusion}
 \label{sec:conclusion}
  
  Summarising, this paper shows that in a highly flexible affine framework, a very general class of insurance policies, including variable annuities with guarantees, can be valued in a tractable and insurance-finance arbitrage-free way. The formulas are more involved, if the payoffs become more complicated, as is the case for example for the death benefit - where in particular interactions between survival and surrender times need to be handled.

\end{document}